\documentclass[pra,twocolumn]{revtex4}

\usepackage{amsmath, amssymb, amsthm, amsfonts, latexsym, graphicx}
\usepackage{float}
\usepackage{verbatim} 
\usepackage[colorlinks]{hyperref}
\usepackage{url}
\usepackage{complexity}
\usepackage{hypcap}
\newcommand{\ket}[1]{\left| #1\right\rangle}        % ket vector
\newcommand{\bra}[1]{\left\langle #1\right|}        % bra vector
\newcommand{\ketbra}[2]{| #1 \rangle\!\langle #2 |} % <x|y>
\newtheorem{definition}{Definition}
\newtheorem{theorem}{Theorem}
\newtheorem{lemma}{Lemma}

\newcommand{\eq}[1]{Eq.~\hyperref[eq:#1]{(\ref*{eq:#1})}}
\renewcommand{\sec}[1]{\hyperref[sec:#1]{Section~\ref*{sec:#1}}}
\newcommand{\app}[1]{\hyperref[app:#1]{Appendix~\ref*{app:#1}}}
\newcommand{\tab}[1]{\hyperref[tab:#1]{Table~\ref*{tab:#1}}}
\newcommand{\fig}[1]{\hyperref[fig:#1]{Figure~\ref*{fig:#1}}}
\newcommand{\figa}[2]{\hyperref[fig:#1]{Figure~\ref*{fig:#1}#2}}
\newcommand{\figx}[2]{\hyperref[fig:#1]{Figure~\ref*{fig:#1}(#2)}}
\newcommand{\thm}[1]{\hyperref[thm:#1]{Theorem~\ref*{thm:#1}}}
\newcommand{\lem}[1]{\hyperref[lem:#1]{Lemma~\ref*{lem:#1}}}
\newcommand{\cor}[1]{\hyperref[cor:#1]{Corollary~\ref*{cor:#1}}}
\newcommand{\defn}[1]{\hyperref[def:#1]{Definition~\ref*{def:#1}}}
\newcommand{\alg}[1]{\hyperref[alg:#1]{Algorithm~\ref*{alg:#1}}}
\newcommand{\prob}[1]{\hyperref[prob:#1]{Problem~\ref*{prob:#1}}}

\begin{document}
\title{Tomography and Generative Data Modeling via Quantum Boltzmann Training}
\author{M\'{a}ria Kieferov\'{a}, Nathan Wiebe}
\begin{abstract}
The promise of quantum neural nets, which utilize quantum effects to model complex data sets, has made their development an aspirational goal for quantum machine learning and quantum computing in general.  Here we provide new methods of training quantum Boltzmann machines, which are a class of recurrent quantum neural network.  Our work generalizes existing methods and provides new approaches for training quantum neural networks that compare favorably to existing methods.  We further demonstrate that quantum Boltzmann machines enable a form of quantum state tomography that not only estimates a state but provides a prescription for generating copies of the reconstructed state.  Classical Boltzmann machines are incapable of this.   Finally we compare small non--stoquastic quantum Boltzmann machines to traditional Boltzmann machines for generative tasks and observe evidence that quantum models outperform their classical counterparts.
\end{abstract}
\date{\today}
\maketitle

\emph{Introduction}-- The Boltzmann machine is a widely used type of recurrent neural net that, unlike the feed forward neural nets used in many applications, is capable of generating new examples of the training data~\cite{hinton2002training}.  This makes it an excellent model to use in cases where data is missing.  We focus on Boltzmann machines because, of all neural net models, the Boltzmann machine is perhaps the most natural one for physicists.  It models the input data as if it came from an Ising model in thermal equilibrium.  The goal of training is then to find the Ising model that is most likely to reproduce the input data which is known as a training set.

The close analogy between this model and physics has made it a natural fit for quantum computing and quantum annealing.  A number of proposals have been put forward for accelerating Boltzmann machines in current generation quantum annealers~\cite{denil2011toward,adachi2015application,benedetti2015estimation} and quantum computers~\cite{wiebe2016quantum}, the latter showing polynomial speedups relative to classical training~\cite{wiebe2015quantum}.  While these methods showed that quantum technologies can train Boltzmann machines more accurately and at lower cost than classical methods, the question of whether transitioning from an Ising model to a quantum model for the data would provide substantial improvements.

This question is addressed in~\cite{amin2016quantum}, wherein a new method for training Boltzmann machines is provided that uses transverse Ising models in thermal equilibrium to model the data.  While such models are trainable and can outperform classical Boltzmann machines, the training procedure proposed therein suffers two drawbacks.  First, it is unable to learn quantum terms from classical data.  Second, the transverse Ising models considered are widely believed to be simulatable using quantum Monte-Carlo methods.  This means that such models are arguably not quantum and as such the benchmarks they give do not necessarily apply to manifestly quantum models.  Here we rectify these issues by giving new training methods that do not suffer these drawbacks and illustrate their performance for models that are manifestly quantum.

The first, and arguably most important, task when approaching the problem of training Boltzmann machines within a quantum setting is to define the model and the problem.  Our approach to quantum Boltzmann machine training requires two inputs in order to specify the training process.  The first is a Hamiltonian model that is used to enforce energy penalties between different states as well as to allow quantum correlations between concepts.  We formally define a Hamiltonian for quantum Boltzmann training as follows.

%For simplicity, we focus our attention on Boltzmann machines with binary units.  Similar definitions for Boltzmann machines with units that 
%take values from $\mathbb{Z}_m$ for positive integer $m>2$ are easy to define similarly.

\begin{definition}\label{def:qBoltz}
Let $V$ be a set of $n$ vertices and $E$ be a set of edges connecting these vertices.
Then define $H\in \mathbb{C}^{2^n\times 2^n}$ to be a Hermitian matrix such that $H=H_{\rm cl}+H_{\rm qm}$ where $H_{\rm cl}$ is the classical Boltzmann model
$$
H_{\rm cl} =  \sum_{j\in V} b_j \hat n_j + \sum_{k \in E} w_k \hat n_{k_1}\hat n_{k_2},
$$
where $\hat n_j =(\openone - \sigma_z^{(j)})/2$ and $H_{\rm qm}$ is a matrix such that $\|H_{\rm qm}-{\rm diag}(H_{\rm qm})\|>0$.
\end{definition}
The second element is the training data.  The training data is what provides the Boltzmann machine with typical patterns that we aim to train it to recognize.  In classical Boltzmann training, the training data (for binary units) comprises of a set of boolean vectors on $n$ bits.  There is a much richer set of training vectors that are permissible in quantum training.

%An important consequence of this model is that while the Boltzmann machine is a graphical model the quantum Boltzmann machine need not be.  In fact, we will see that there are many natural examples of quantum systems that can be used as models for learning that have interactions that cannot be described in terms of an interaction graph.

\emph{Training Quantum Boltzmann Machines}--  Examining the quantum analogue of machine learning algorithms is very much like examining the 
quantum analogue of classical dynamical systems: there are many ways that the classical Boltzmann machine can be translated to a quantum 
setting.  We propose two methods that we refer to as POVM based training and state based training.  The basic 
difference between the two approaches stems from the correspondence used to generalize the notion of classical training data into a quantum 
setting.

POVM based training is a generalization of the approach in the work of Amin et al~\cite{amin2016quantum}, which assumes that the user is provided with a 
discrete set of training vectors that are assumed to be sampled from an underlying training distribution.  In the case of the prior work, 
the algorithm is trained using projectors on the classical training states.  The goal of training is then to find the quantum Hamiltonian 
that maximizes the log--likelihood of generating the observed training vectors.  Here we describe the training set as a set of measurement 
and generalize it to allow the measurement record to 
correspond to labels of POVM elements. In the more general approach, the training data can correspond to density operators as 
well as non-orthogonal states.

\begin{definition}
Let $\mathcal{H}:=\mathcal{V}\otimes \mathcal{L}$ be a finite-dimensional Hilbert space describing a quantum Boltzmann machine and let $\mathcal{V}$ and $\mathcal{L}$ be subsystems corresponding to the visible and latent units of the QBM.
The probability distribution $P_v$ and POVM $\Lambda = \{\Lambda_v\}$,  comprise a training set for QBM training if 1) there exists a bijection between the domain of $P_v$ and $\Lambda$ and 2) the domain of each $\Lambda_v$  is $\mathcal{H}$ and it acts non--trivially only on subsystem $\mathcal{V}$.
\end{definition}

As a clarifying example, consider the following training set.  Let us imagine that we wish to train a model that generates even numbers.  Then a sensible training set would be
\begin{align}
\Lambda_n &= \ketbra{2n}{2n} \text{ for } 1\leq n\leq 8\label{eq:Lambdaex}\\
\Lambda_0 &= \openone - \sum_{n=1}^{8} \Lambda_n,\qquad
P_v=(1-\delta_{v,0})/8.
\end{align}
The following equivalent training set can also be used
\begin{align}
\Lambda_1 &= \frac{1}{{8}} \big(\ket{2} + \cdots +\ket{16}\big)\big(\bra{2} + \cdots +\bra{16}\big), \\
\Lambda_0 &= \openone - \Lambda_1,\qquad
P_v = \delta_{v,1}.
\end{align}
Both learning problems aim to mimic the same probability distributions.
This shows that the training data for quantum Boltzmann training can be complex even when a single 
 training vector is used.  
%We will see later that this observation is key to allowing some forms of POVM based training to learn 
%some classes of Hamiltonian models.  

The second approach assumes that the training data is provided directly through a quantum state that gives the true distribution over the data.  This approach is typically stronger than POVM based training because measurements of the state can provide the statistics needed to perform the former form of training.  We will see that it has an advantage in that it can easily allow quantum Boltzmann machines to perform a type of tomographic reconstruction of the state and also can allow forms of training that make fewer approximations than existing methods.  The former advantage is particularly significant as it creates a link between quantum state tomography and quantum machine learning. We define the training set as follows.

\begin{definition}
Let $\mathcal{H}$ be a finite--dimensional Hilbert space and let $\rho$ be a Hermitian operator on $\mathcal{H}$. The operator $\rho$ is a training set for state based training if it is a density operator.
\end{definition}

As a clarifying example, the training data given in~\eq{Lambdaex} could correspond to the following training data for state based learning
\begin{equation}
\rho = \frac{1}{8}\left(\ketbra{2}{2} + \cdots + \ketbra{16}{16} \right)
\end{equation}
In state based training we assume that copies of $\rho$ are prepared by an oracle and will not assume that the user has neither performed any experiments on $\rho$ nor has any prior knowledge about it.  This is in contrast to POVM based training where the user has a set of measurement records but not the distribution it was drawn from.

\emph{Hamiltonian}--
The last part to specify is the Hamiltonian for a quantum Boltzmann machine.  There are many Hamiltonians that 
one could consider.  Perhaps the most natural extension to the Boltzmann machine is to consider the transverse Ising model, which was 
investigated by Amin et al~\cite{amin2016quantum}.  Here we consider a different example motivated by the fact that the stoquastic Hamiltonians used in 
previous work can be efficiently simulated using quantum Monte-Carlo methods.  In order to combat this, we explicitly consider Hamiltonians 
that are Fermionic because the Fermionic sign problem prevents quantum Monte-Carlo methods from providing an efficient simulation.

\begin{definition}\label{def:qBoltz}
Let $V$ be a set of $n$ vertices and $E$ be a set of edges connecting these vertices.
Then define $H\in \mathbb{C}^{2^n\times 2^n}$ to be a Hermitian matrix such that $H=H_{\rm cl}+H_{\rm qm}$ where $H_{\rm cl}$ is the classical Boltzmann model
$$
H_{\rm cl} =  \sum_{j\in V} b_j \hat n_j + \sum_{k \in E} w_k \hat n_{k_1}\hat n_{k_2},
$$
where $\hat n_j =(\openone - \sigma_z^{(j)})/2$ and $H_{\rm qm}$ is a matrix such that $\|H_{\rm qm}-{\rm diag}(H_{\rm qm})\|>0$.
\end{definition}

The Hamiltonian we consider is of the form
\begin{equation}
H = H_{p} + \frac{1}{2} H_{pq} + \frac{1}{2} H_{pqrs},
\end{equation}
where $H_{p}= \sum_p h_{p} \left( a_p + a_p^{\dagger} \right)$, $H_{pq} = \sum_{pq}h_{pq} \left(  a_p^{\dagger} a_q +  a_q^{\dagger} a_p  \right)$ and $H_{pqrs} = \sum_{pqrs}h_{pqrs} (a_p^{\dagger}a_q^{\dagger} a_r a_s +h.c.)$
%\begin{align}
%H_{p}&= \sum_p h_{p} \left( a_p + a_p^{\dagger} \right),\\
%H_{pq} &= \sum_{pq}h_{pq} \left(  a_p^{\dagger} a_q +  a_q^{\dagger} a_p  \right),\\
%H_{pqrs} &= \sum_{pqrs}h_{pqrs} (a_p^{\dagger}a_q^{\dagger} a_r a_s +h.c.).
%\end{align}
Here $a_p$ and $a^\dagger_p$ are Fermionic creation and annihilation operators, which create and destroy Fermions at unit $p$.  They have the properties that
%\begin{align}
$a^\dagger \ket{0}=\ket{1}$,
$a^\dagger \ket{1}=0$ and
$a^\dagger_p a_q + a_qa^\dagger_p = \openone\delta_{pq}.$
%\end{align}
The Hamiltonian here corresponds to the standard Hamiltonian used in quantum chemistry modulo the presence of the non--particle conserving $H_p$ term.
Note that all terms in the Hamiltonian conserve the number of Fermions with the exception of $H_{p}$.  We include this term to allow the distribution to have superpositions over different numbers of Fermions which is necessary for this model to be able to learn generative models for certain classes of pure states.

\begin{figure*}

\begin{minipage}{0.49\linewidth}
\includegraphics[width=\textwidth]{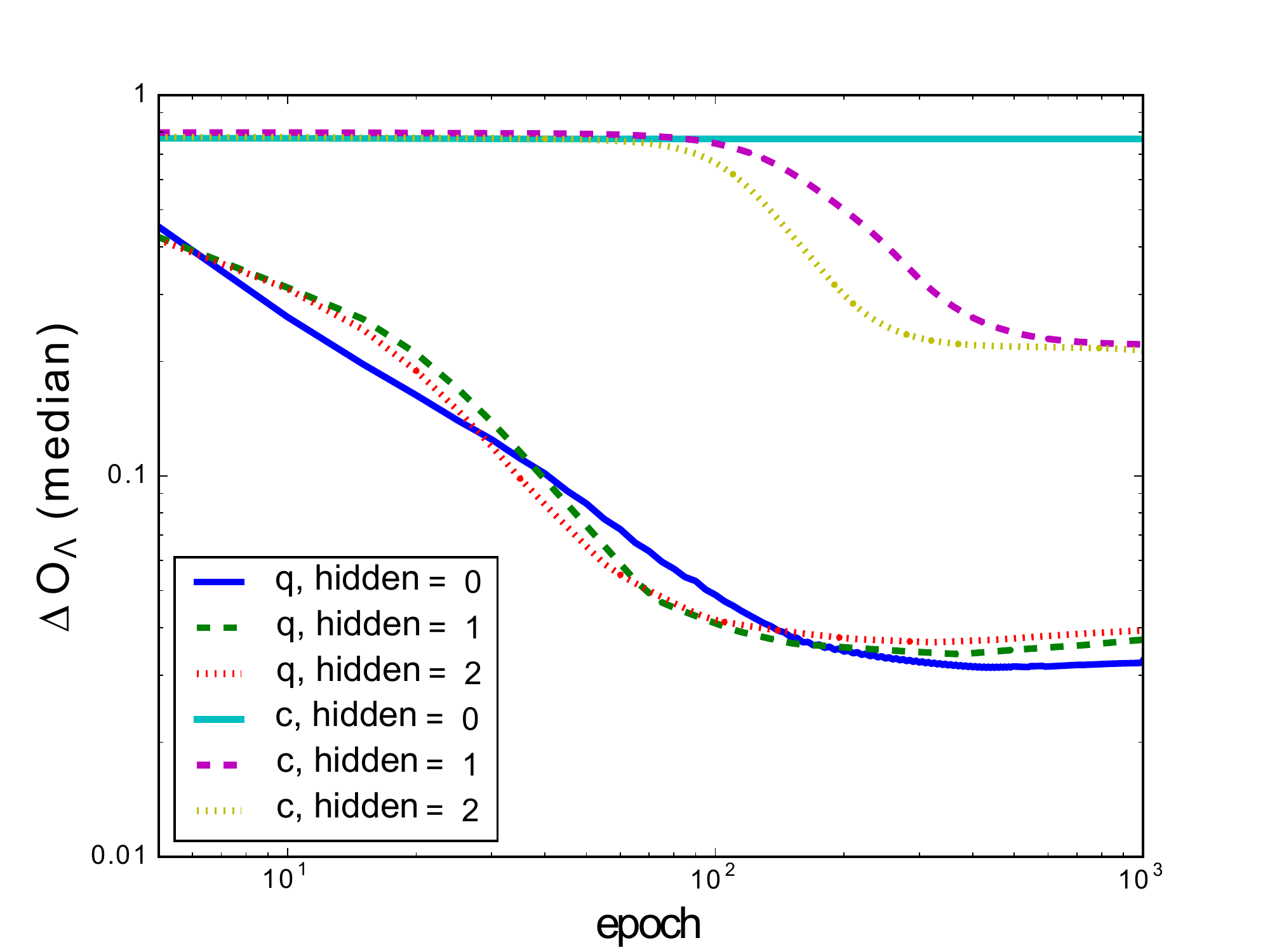}\label{5}
(a)
\end{minipage}
\hspace{1mm}
\begin{minipage}{0.49\linewidth}
\includegraphics[width=\textwidth]{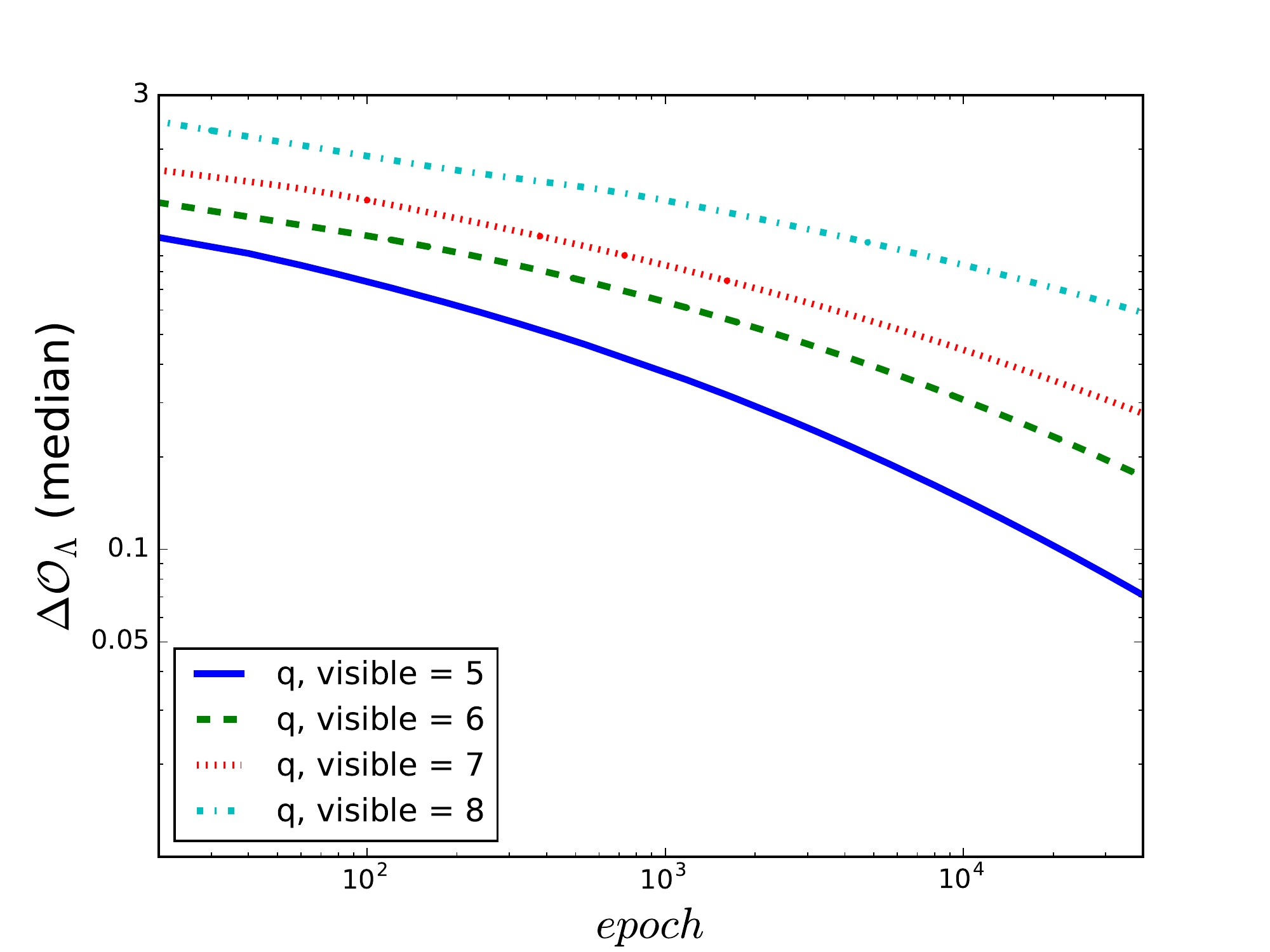}\label{fig:visible}
(b)
\end{minipage}
\caption{Simulation of QBM with POVM training. Subfigures (a)  and (b) where we compute $\Delta\mathcal{O}_{\Lambda}:= \mathcal{O}_{\Lambda,\max} - \mathcal{O}_{\Lambda}$ for (a) $5$ visible units and varying numbers of hidden units and (b) for all relative entropy training with all visible Boltzmann machines.  We take $\lambda=0$ for all data considered and $\max\mathcal{O}_{\Lambda,\max}$ is the maximum value of the training objective function attainable for the training data.}
\label{fig:4plots}
\end{figure*}

By inspection it is clear from the fact that $n_p =a^\dagger_pa_p$ that this Fermionic 
Hamiltonian reduces to the Ising model used in Boltzmann training if we set $H_{p}\rightarrow 0$ and take all other non--diagonal terms in 
$H$ to be $0$.  Therefore, apart from being difficult to simulate using quantum Monte-Carlo, this model of Fermionic Boltzmann machines 
encompasses traditional Boltzmann training while at the same time being expressible as $H=H_{\rm qm} + H_{\rm cl}$ as per~\defn{qBoltz}.

\emph{Golden--Thompson training}--  The training process for a quantum Boltzmann machine can be viewed as optimizing an objective function that measures how close the distribution generated from the model is to the underlying data distribution.  The goal is to modify the Hamiltonian parameters to maximize the objective function given the data collected and the Hamiltonian model.  

We consider two different forms of the objective function here corresponding to POVM based training and state based training.  The first, 
and simplest, objective function that we will discuss is that corresponding to POVM based training.  The goal of this form of training is to 
minimize the KL-divergence between the empirically observed data distribution in the training set and that produced by the Hamiltonian model 
in thermal equilibrium.  Let us define our POVM to be $\Lambda := \{\Lambda_0,\ldots, \Lambda_N\}$ and further define
\begin{equation}
\mathcal{L}:=P(v|H) := {\rm Tr}\left[ \Lambda_v  \frac{e^{-H}}{{\rm Tr} [e^{-H}]}\right],
\end{equation}
where $\openone_h$ is the identity operator on the hidden units of the model.  The KL divergence is then
\begin{equation}
{\rm KL}(P\|\mathcal{L})=\sum_v P_v \log\left(\frac{P_v}{P(v|H)} \right),
\end{equation}
and since $P_v$ is a constant, minimizing this objective function is equivalent to maximizing $\sum_v P_v \log(P(v|H))$.  The latter term is known as the average log--likelihood.  The objective function that we then wish to optimize is
\begin{equation}
\mathcal{O}_{\Lambda}(H;\lambda) = \sum_{\mathbf{v}} P_{\mathbf{v}} \log\left({\frac{{\rm Tr}\left[ \Lambda_v e^{-H}\right]}{{\rm Tr}\left[ e^{-H} \right]}}\right)-\frac{\lambda}{2} \|h_Q\|^2.
\end{equation}
where $h_Q$ is the vector of Hamiltonian terms that correspond to off-diagonal matrix elements.  The last term is an $L_2$ regularization term that we include to penalize quantum terms that are not needed to explain the data.

While this objective function is unlikely to be generically computable because the calculation of ${\rm Tr}[e^{-H}]$ is a $\#\P$--hard problem, the gradients of the objective function are not hard to estimate for classical Boltzmann machines.  A challenge emerges for quantum Boltzmann machines because $H$ need not commute with its parametric derivatives.  In particular, let $\theta$ be a Hamiltonian parameter then Duhamel's formula gives 
\begin{equation}
{\rm Tr}\left[\Lambda_v \partial_\theta e^{-H}\right]={\rm Tr}\left[\int_0^1 \Lambda_v e^{sH} \left[\partial_\theta H\right] e^{(1-s)H} \mathrm{d}s\right].\label{eq:duhamel}
\end{equation}
If $\Lambda_v$ commuted with $H$, then we would recover an expression for the gradient that strongly resembles the classical case.  

A solution to this problem was proposed in~\cite{amin2016quantum}, wherein the Golden--Thompson inequality is used to optimize a lower bound to the objective function.  Upon using this inequality we find the following expression for the derivative of the objective function.

\begin{align}
& \sum_{\mathbf{v}} P_{\mathbf{v}} \log\left({\frac{{\rm Tr}\left[ \Lambda_v e^{-H}\right]}{{\rm Tr}\left[ e^{-H} \right]}}\right)-\frac{\lambda}{2} \|h_q\|^2\nonumber\\
&\qquad\ge   \sum_{\mathbf{v}} P_{\mathbf{v}} \log\left({\frac{{\rm Tr}\left[  e^{-H+\log(\Lambda_v)}\right]}{{\rm Tr}\left[ e^{-H} \right]}}\right)-\frac{\lambda}{2} \|h_q\|^2,
\end{align}
 which leads to to an expression analogous to the classical training.   This inequality is saturated when $[\Lambda_v,H_v]=0$, from the Baker-Campbell-Hausdorff formula.
  The gradient of the lower bound on the objective is
\begin{equation}
 \sum_{\mathbf{v}}P_{\mathbf{v}} \Big( -\frac{\text{Tr}[e^{-H_v} \partial_{\theta} H]} {\text{Tr}[e^{-H_v}] } + \frac{\text{Tr}[e^{-H}\partial_{\theta}H]} {\text{Tr}[e^{-H}] }  \Big)-\lambda h_{\theta} \delta_{H_\theta \in H_Q},
\end{equation}
where  $H_v = H -\log \Lambda_v $.  We will examine this form of training below and find that it yields excellent gradients in the cases considered and agrees well with the exact expectation values, as we show in the appendix.

This form of training is incapable of learning any component of the Hamiltonian such that $\text{Tr}[e^{-H_v} \partial_{\theta} H]=0$.
This meant that the weights corresponding to quantum terms could not be trained directly in previous work that only considered the training data to arise from $\Lambda_v = \ketbra{y_v}{y_v}$ where $y_v$ is the binary representation of the $v^{\rm th}$ training vector.  These components instead had to be guessed, which becomes impractical as the number of independent quantum terms grows.  Here we eschew such approaches by allowing $\Lambda_v$ to contain non--diagonal POVM elements.

\emph{Relative Entropy Training}--The second approach that we will consider is optimizing the relative entropy instead of the average 
log--likelihood.  In this case, the objective function that we wish to optimize is
\begin{equation}
\mathcal{O}_{\rho}(H;\lambda) = S(\rho\| e^{-H}/{\rm Tr}[e^{-H}])-\frac{\lambda}{2} \|h_Q\|^2,
\end{equation}
and the derivatives of the objective function are
\begin{equation}
-{\rm Tr}[\rho \partial_\theta H] + {\rm Tr}[e^{-H}\partial_\theta H]/{\rm Tr}[e^{-H}]-\lambda h_{\theta} \delta_{H_\theta \in H_Q}.
\end{equation}
Thus we can systematically make the state generated by a simulator of $e^{-H}/Z$ harder to distinguish from the state $\rho$ by following a gradient given by the difference between expectations of the Hamiltonian terms in the data distribution $\rho$ and the corresponding expectation values for $e^{-H}/Z$.

$\mathcal{O}_{\rho}$ is motivated by the fact that $S(\rho||e^{-H}/Z)\ge \|\rho -e^{-H}/Z\|^2/2\ln(2)$ if $\rho$ is positive definite. Thus if $\rho$ has maximum rank, $S(\rho||e^{-H}/Z)\rightarrow 0$ implies $e^{-H}/Z \rightarrow \rho$. 

There are two major advantages to this method.  The first is that no approximations are needed to compute the gradient.  The second is that this directly enables a form of tomography wherein a model for the state $\rho$ is provided by the Hamiltonian learned through the gradient ascent process.  This is further interesting because this procedure is efficient, given that an accurate approximation to the thermal state can be prepared for $H$, and thus it can be used to describe states in high dimensions.  The procedure also provides an explicit procedure for generating copies of this inferred state, unlike conventional tomographic methods.

%The main disadvantage of this method are that the gradients do not hold precisely for systems with hidden units.

\begin{figure}[t!]
\includegraphics[width=\linewidth]{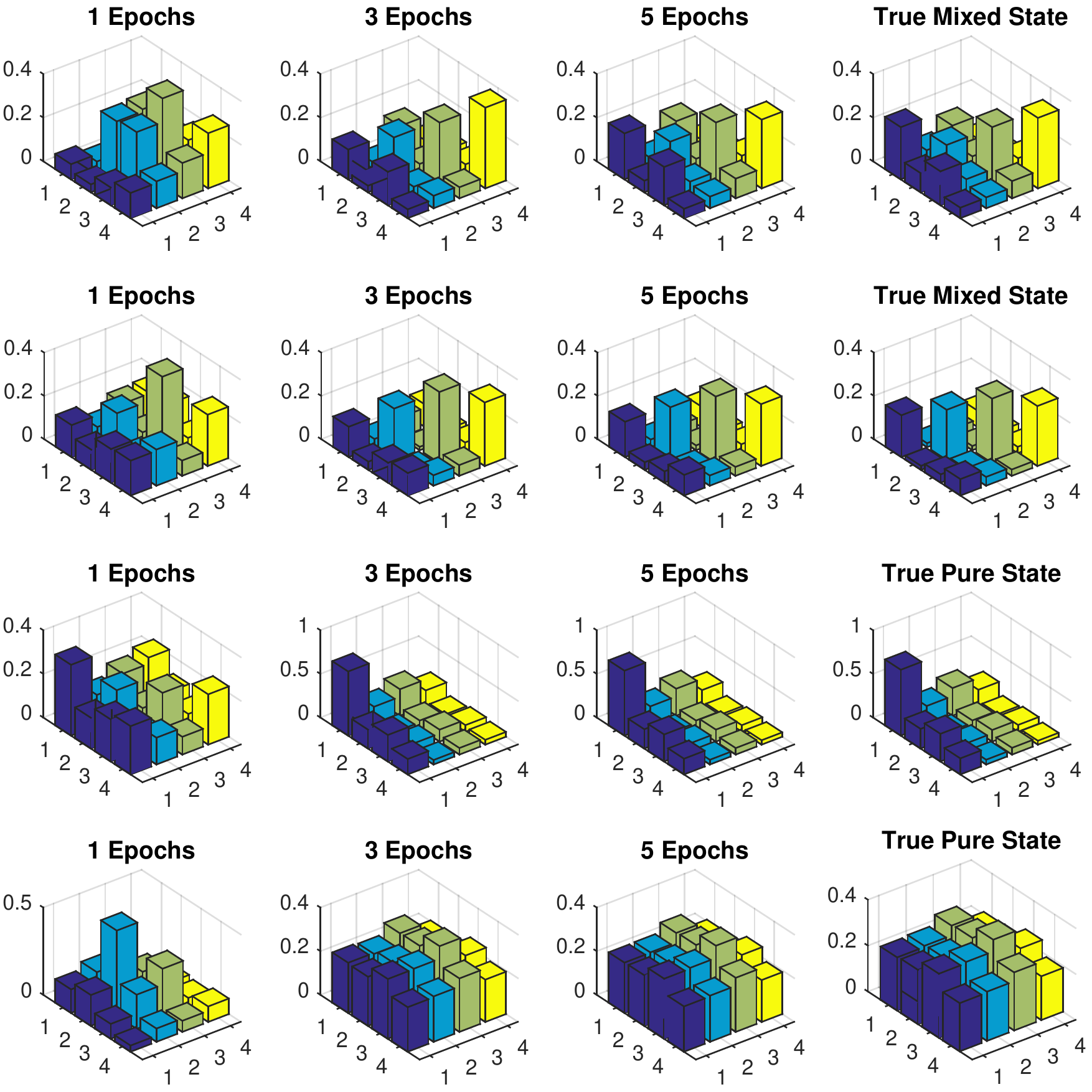}
\caption{Absolute values of tomographic reconstructions of random two--qubit mixed and pure states using relative entropy training with $\eta=1$.}\label{fig:tomo}
\end{figure}

\emph{Numerical Results}--We first examine the performance of our algorithm for generative training using a small number of visible and hidden units and compare the result to classical training.  Since we can only simulate small quantum computers classically, we choose a simple training set composed of step functions where the step occurs at each possible value with $10\%$ noise added to the vectors.  We assume that $\Lambda_1 = \ketbra{\psi}{\psi}$, $\Lambda_0 = \openone - \ketbra{\psi}{\psi}$ for POVM training where $\ket{\psi}$ is a pure state constructed in the above fashion.  We similarly take $\rho = \ketbra{\psi}{\psi}$ for relative entropy training.

The data in~\fig{4plots} shows that, in every instance examined, the quantum model does a better job fitting the data than the classical model does. 
We observe that increasing the number of hidden units gives the classical methods a substantial advantage, but we do not notice that adding hidden units substantially improves $\mathcal{O}_{\Lambda}$ here.  This is likely because the training data is sufficiently simple that the $H_{pqrs}$ terms give enough freedoms to fit the data without the need for hidden units.  

When we consider relative entropy training, we note that the value of the objective function seems to systematically grow with the size of the Boltzmann machine.  This is expected because the complexity of the training data grows as we increase the number of visible units.  We see that qualitatively that training continues to improve the value of the objective function here and given the computational resources at our disposal, we were unable to see the learning stop despite training with gradients of the relative entropy objective rather than those of the reported objective function $\mathcal{O}_{\Lambda}$.

This data, along with further data in the supplementary material, suggests that Fermionic quantum Boltzmann machines may be superior models for data; however, further study is needed to ensure that these models are not overfitting the data.

We demonstrate the ability of Quantum Boltzmann training to perform tomography by learning ensembles of two--qubit states that are either Haar--random pure 
states or mixed states that are convex combinations of columns vectors of Haar-random unitary matrices with uniformly distributed weights.  
Here for simplicity we choose our Hamiltonian to consist of every $2$--qubit Pauli operator.  Since this set is complete, every possible 
state can be generated using an appropriate Hamiltonian.  We provide data to this effect in~\fig{tomo}, wherein as few as 5 training epochs suffice 
to learn these states within graphical accuracy.  We provide further details of the error versus epoch tradeoff in the supplementary 
material.

%\emph{ Preparing the thermal distribution} An essential partion of Boltzmann machine training is sampling from the thermal distribution. 
%Sadly, 
%preparing the thermal state is NP-hard. Classical algorithms circumvent this problem by approximating it using contrastive divergence. We 
%can recreate this method in quantum setting using result from []

%A high precision approximation can be obtained using a methods form [Somma]. The authors give a method for constructing a unitary $V$ 
%such that 
%\begin{equation}
%\frac{1}{2}|| Tr_a \left[ V \left(\ketbra{0}{0} \otimes \ketbra{0}{0}_a\right) V^{\dagger} \right] - \frac{e^{-\beta H}}{Z} ||_1 < \epsilon.
%\end{equation}

\emph{Conclusion}--  We have investigated a new approach to training quantum Boltzmann machines that does not suffer from the drawbacks presented by previous schemes.  In particular, we see that we can learn both quantum and classical terms in the Hamiltonian through either our POVM-based Golden--Thompson training approach or by training according to the relative entropy.  We then see that this latter example enables a form of tomography, which allows in concert with a quantum simulator for Hamiltonian models to be learned for complex quantum states that cannot be probed using conventional tomographic approaches.

While our work demonstrates the viability of quantum Boltzmann training for broad classes of non-stoquastic Hamiltonians, subsequent work will be needed to establish whether it provides generalizes classical data than classical Boltzmann machines do.  This will be necessary to understand the extent to which quantum models are prone to overfit the data.

Perhaps the most exciting avenue of future work investigated here is the strong link between tomography and quantum machine learning that we establish.  It is our hope that even more powerful and efficient methods for probing quantum systems will be found by combining ideas from quantum machine learning, quantum Hamiltonian learning and tomography within a single protocol.

\begin{acknowledgements}
We would like to thank K.M. Svore, A. Kapoor, F. Brandao and V. Kliuchnikov for useful comments and feedback.
\end{acknowledgements}

\appendix

\section{Preparing Thermal States} An essential part of Boltzmann machine training is sampling from the thermal distribution. Sadly, 
preparing the thermal state is NP-hard. Classical algorithms circumvent this problem by approximating it using contrastive 
divergence~\cite{hinton2002training}. Analogous quantum solutions have been 
proposed in~\cite{poulin2009sampling,yung2012quantum,wiebe2015quantum}.
A high precision approximation can be obtained using
the methods from~\cite{chowdhury2016quantum}. 

The method of Chowdhury and Somma is strongly related to the methods in~\cite{poulin2009sampling,yung2012quantum,wiebe2016quantum}.
The main difference between these methods is that their approach uses an integral transformation to allow the exponential to be approximated
as a linear combination of unitaries.  These operators are then simulated using Hamiltonian simulation ideas as well as ideas from simulating fractional queries.
The complexity of preparing a Gibbs state $\rho\in \mathbb{C}^{N\times N}$ within error $\epsilon$, as measured by the $2$--norm, is from~\cite{chowdhury2016quantum}
\begin{equation}
O\left(\sqrt{\frac{N}{Z}}{\rm polylog}\left(\frac{1}{\epsilon} \sqrt{\frac{N}{Z}} \right) \right),\label{eq:somma}
\end{equation}
for inverse temperature $\beta=1$ and cases where $H$ is explicitly represented as a linear combination of Pauli operators.  This is roughly quadratically better than existing approaches for preparing general Gibbs states if constant $\epsilon$ is required, but constitutes an exponential improvement if $1/\epsilon$ is large.  This approach is further efficient if $Z\in \Theta(N/{\rm polylog}(N))$.  This is expected if roughly a constant fraction of all eigenstates have a meaningful impact on the partition function.  While this may hold in some cases, particularly in cases with strong regularization~\cite{wiebe2016quantum,wiebe2015quantum}, it is not expected to hold generically.

An alternative method for preparing thermal states is proposed by Yung and Aspuru-Guzik.  The approach works by using a Szegedy walk operator whose transition amplitudes are given by the Metropolis rule based on the energy eigenvalue difference between the two states.  These eigenvalues are computed using phase estimation.  A coherent analogue of the Gibbs state is found by using phase estimation on the walk operator, $W$, that follows these transition rules.  The number of applications of controlled $W$ required in the outer phase estimation loop is

\begin{equation}
O\left(\frac{\| H\|^2 }{\epsilon\sqrt{\delta}}\log\left(\frac{\|H\|^2}{\epsilon^2} \right) \right),\label{eq:AAG}
\end{equation}
where $\delta$ is the gap of the transition matrix that defines the quantum walk, $\epsilon$ is the error in the preparation of the thermal state.  Since each application of the walk operator requires estimation of the eigenvalues of $H$, this complexity is further multiplied by the complexity of the quantum simulation.  Provided that the Hamiltonian is a sum of at most $m$ one--sparse Hamiltonians with efficiently computable coefficients then the cost is multiplied by a factor of $m\log(m)/\log\log(m)$ to $m^{2+o(1)}$ depending on the quantum simulation algorithm used within the phase estimation procedure.

These features imply that it is not clear apriori which algorithm is preferable to use for preparing thermal states.  For cases where the partition function is expected to be large or highly accurate thermal states are required,~\eq{somma} is preferable.  If the spectral gap of the transition matrix is small, quantum simulation is inexpensive for $H$ and low precision is required then~\eq{AAG} will be preferable.

%<<<<<<< HEAD
%\begin{figure}[t!]
%\includegraphics[width=\linewidth]{commutator.pdf}
%\caption{Plot showing the efficacy of commutator training for Boltzmann machines with 4 visible and no hidden units. The top 
%lines depict training with Golden-Thompson at 
%first and then switching to commutator training where we see a sudden increase in accuracy. We picked the parameters such that the 
%commutator training is stable; the learning rate is 0.1 and there is no momentum. The bottom line (square) shows the performance of 
%Golden-Thompson training with learning rate 0.4 that is more suitable for this size of the problem.}
%\end{figure}
%=======
%A high precision approximation can be obtained using a methods form [Somma]. The authors give a method for constructing a unitary $V$ 
%such that 
%\begin{equation}
%\frac{1}{2}|| Tr_a \left[ V \left(\ketbra{0}{0} \otimes \ketbra{0}{0}_a\right) V^{\dagger} \right] - \frac{e^{-\beta H}}{Z} ||_1 < \epsilon.
%\end{equation}
%>>>>>>> bf986c37d622b4d3736b3e36196d3db16e83cd6e

\begin{figure*}
\begin{minipage}{0.49\linewidth}
\includegraphics[width=\textwidth]{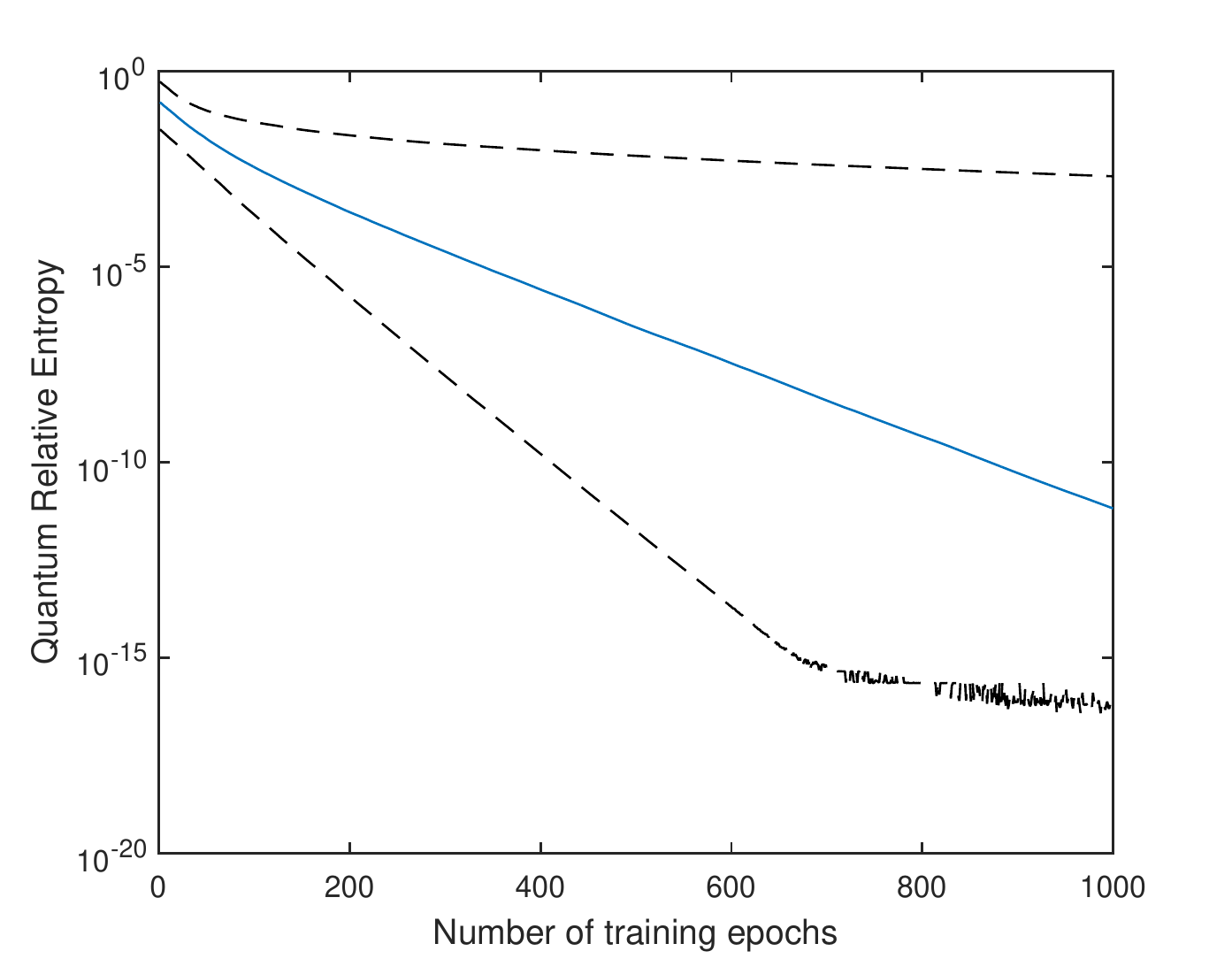}
\end{minipage}
\hspace{1mm}
\begin{minipage}{0.49\linewidth}
\includegraphics[width=\textwidth]{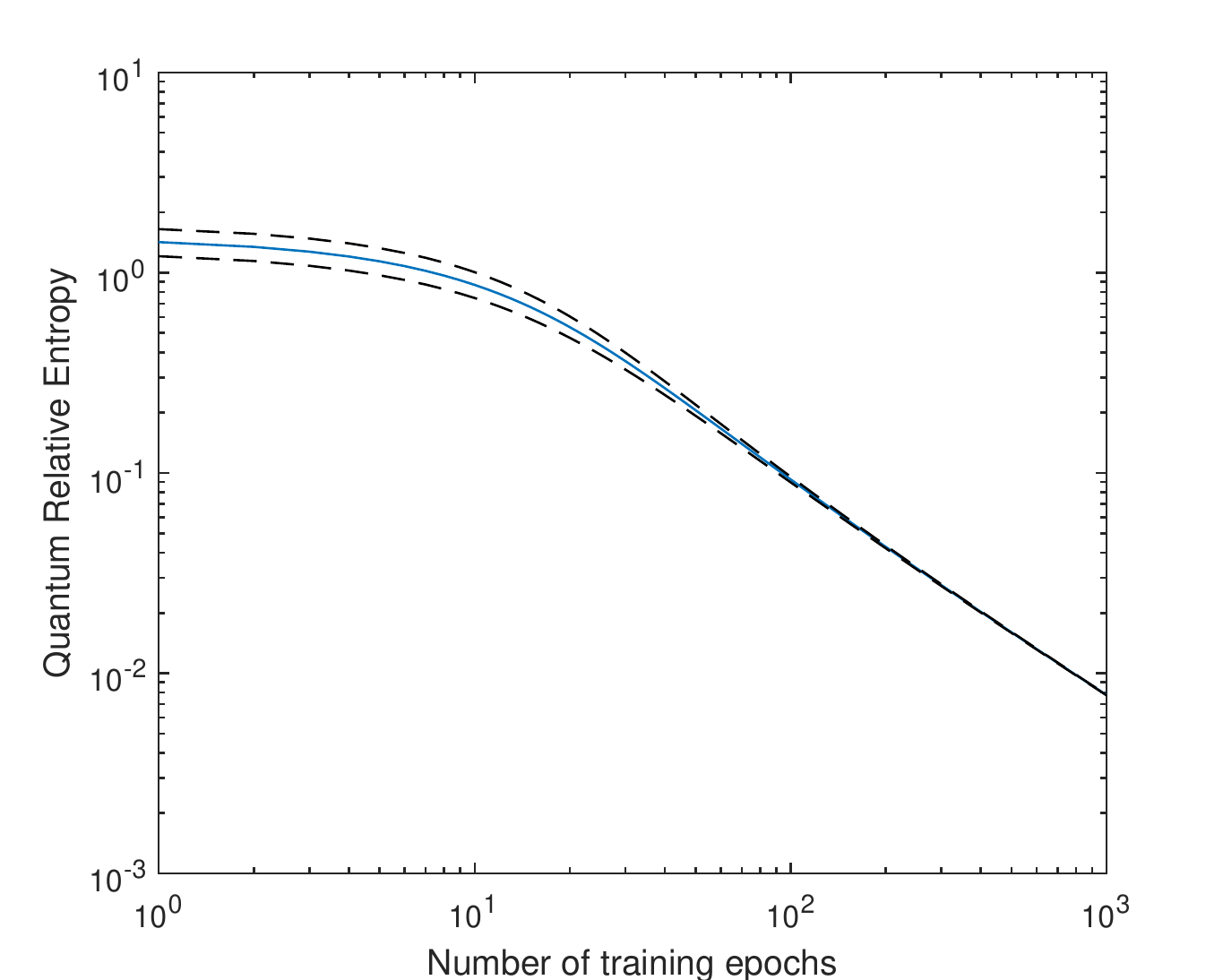}
\end{minipage}
\begin{minipage}{0.49\linewidth}
\includegraphics[width=\textwidth]{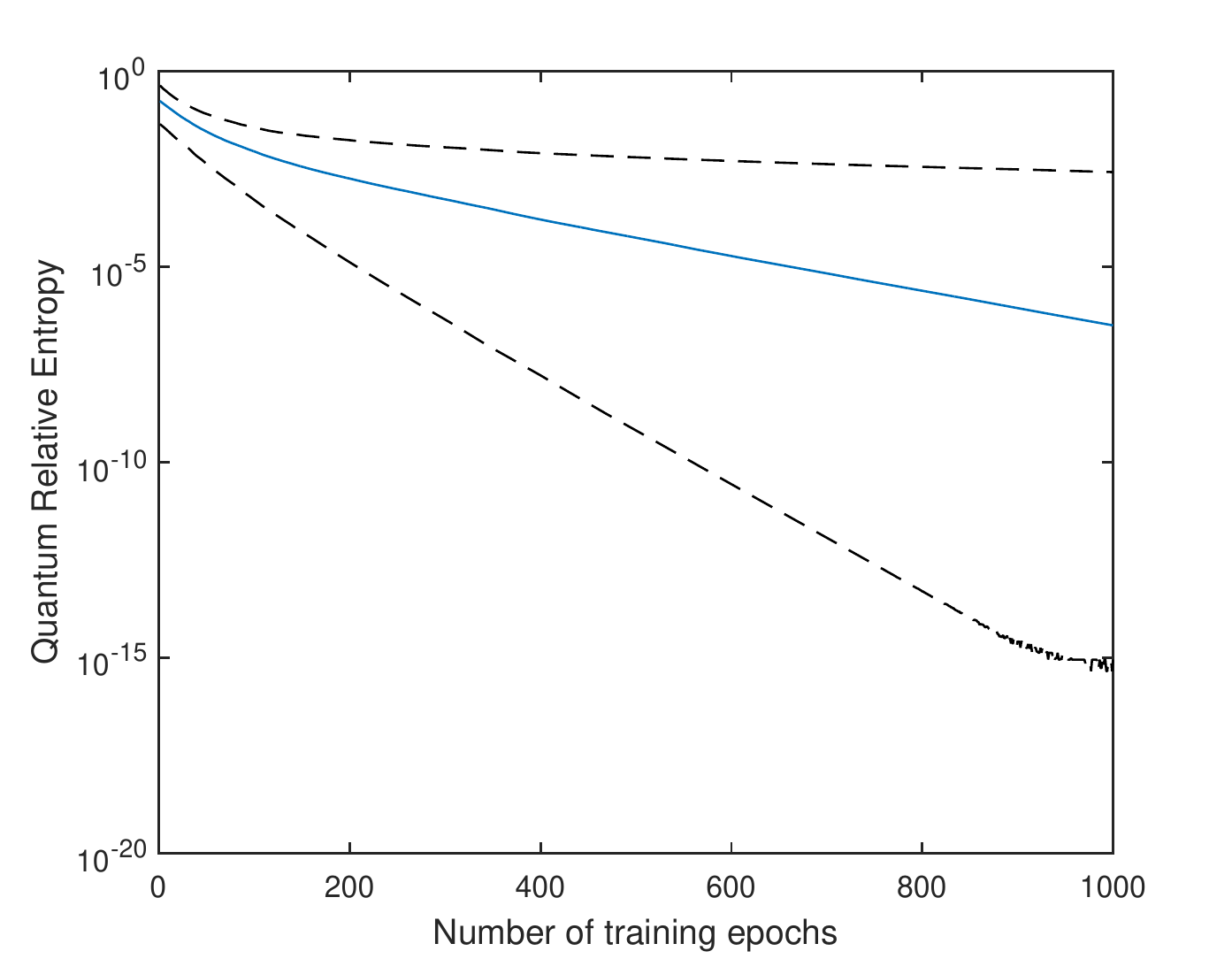}
\end{minipage}
\hspace{1mm}
\begin{minipage}{0.49\linewidth}
\includegraphics[width=\textwidth]{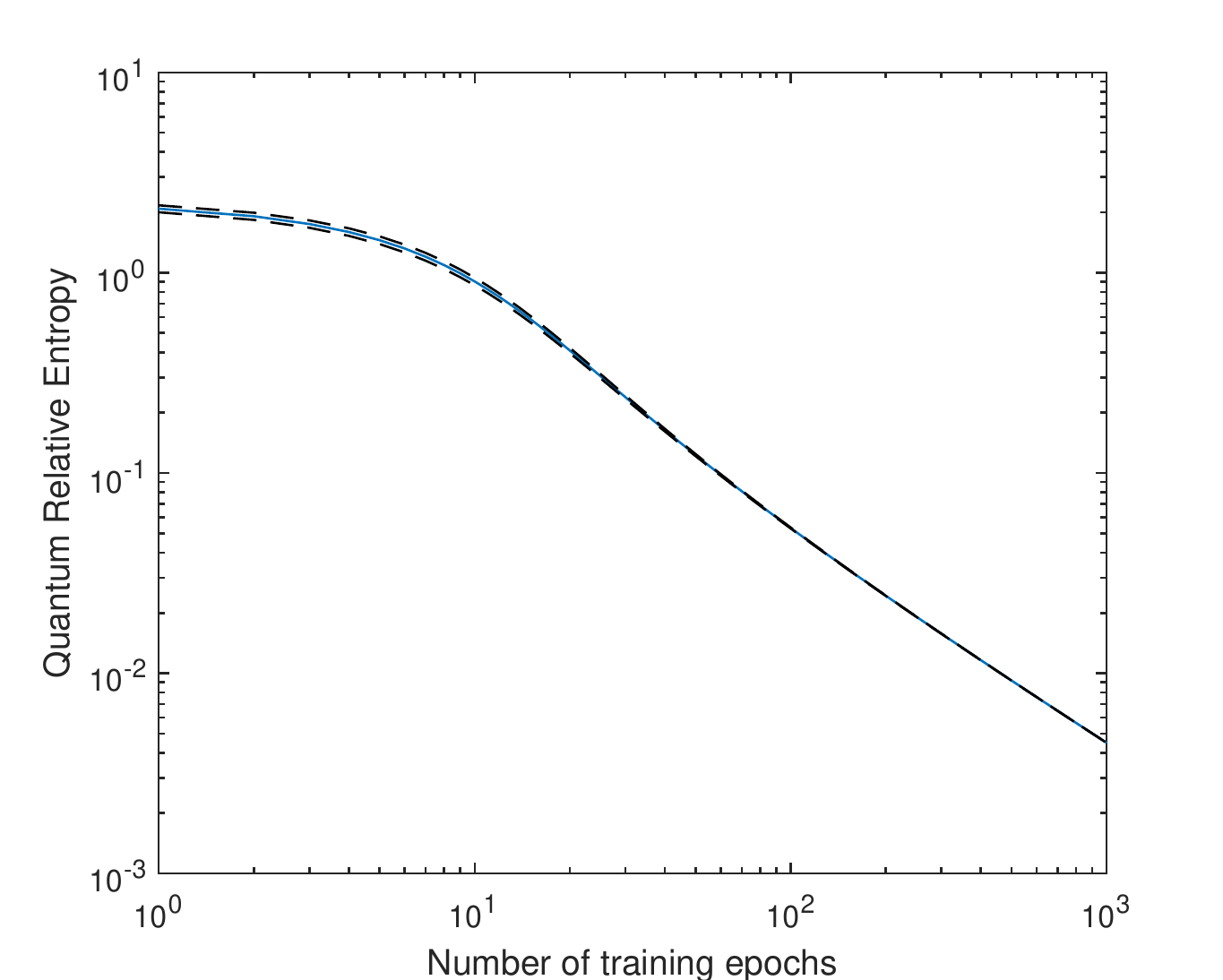}
\end{minipage}
\begin{minipage}{0.49\linewidth}
\includegraphics[width=\textwidth]{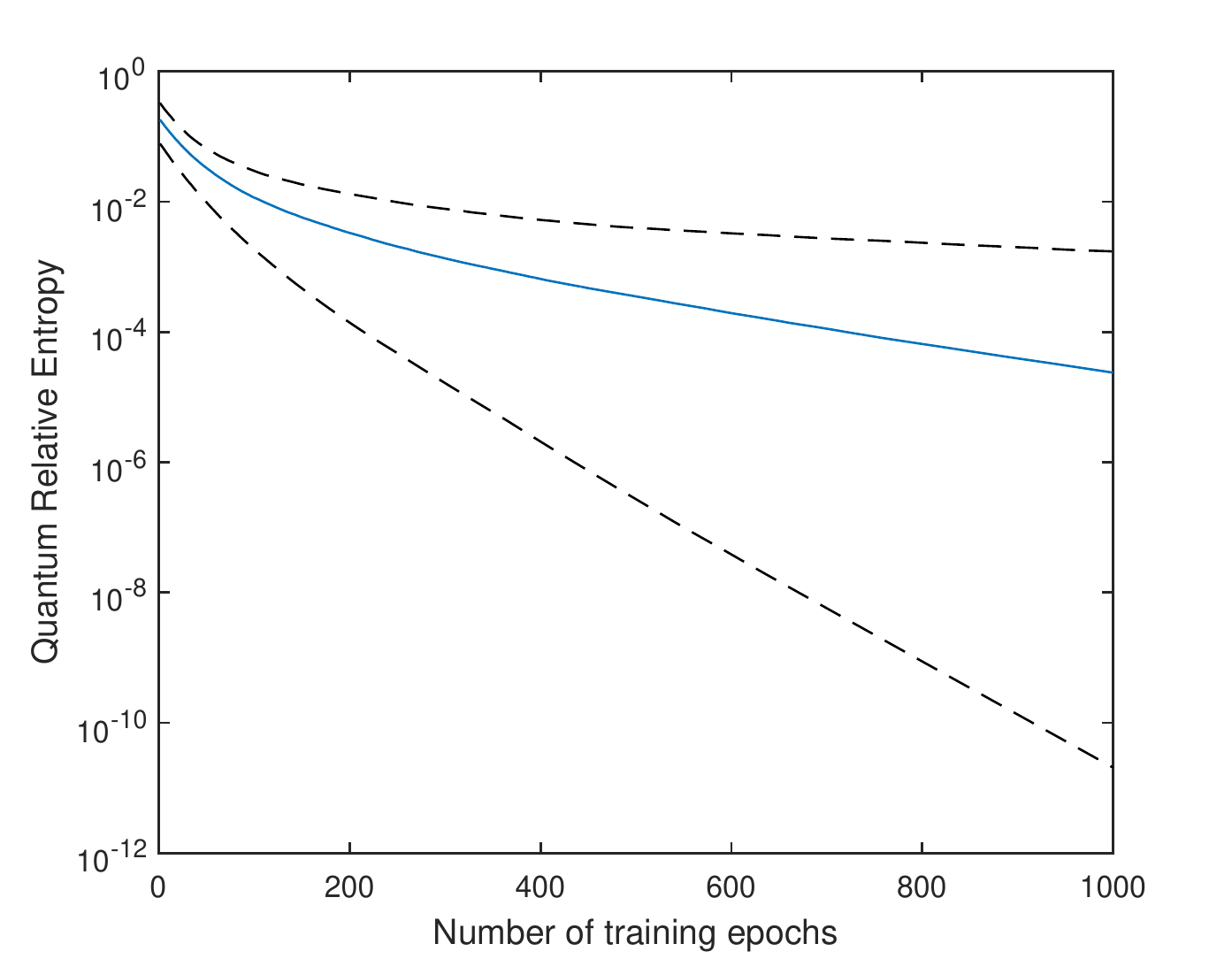}
\end{minipage}
\hspace{1mm}
\begin{minipage}{0.49\linewidth}
\includegraphics[width=\textwidth]{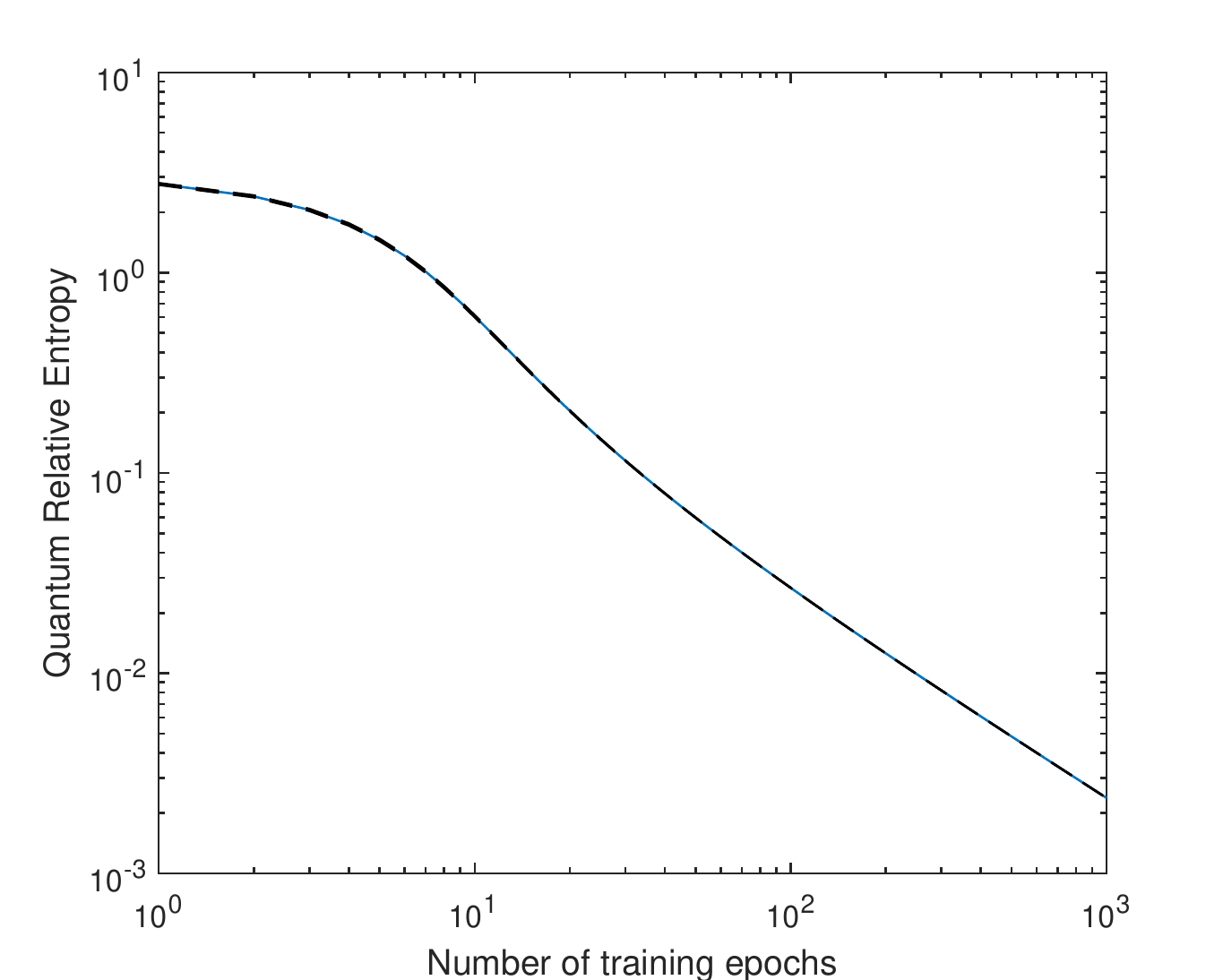}
\end{minipage}
\caption{Distribution of quantum relative entropies between randomly chosen mixed (left) and pure (right) states as a function of the number of training epochs for $2$- (top), $3$- (middle) and $4$- (bottom) qubit tomography with $\eta=0.025$. Dashed lines represent a $90\%$ confidence interval and the solid line denotes the median.\label{fig:relPurity}}
\end{figure*}

\begin{figure*}[t!]
\begin{minipage}{0.49\linewidth}
\includegraphics[width=\textwidth]{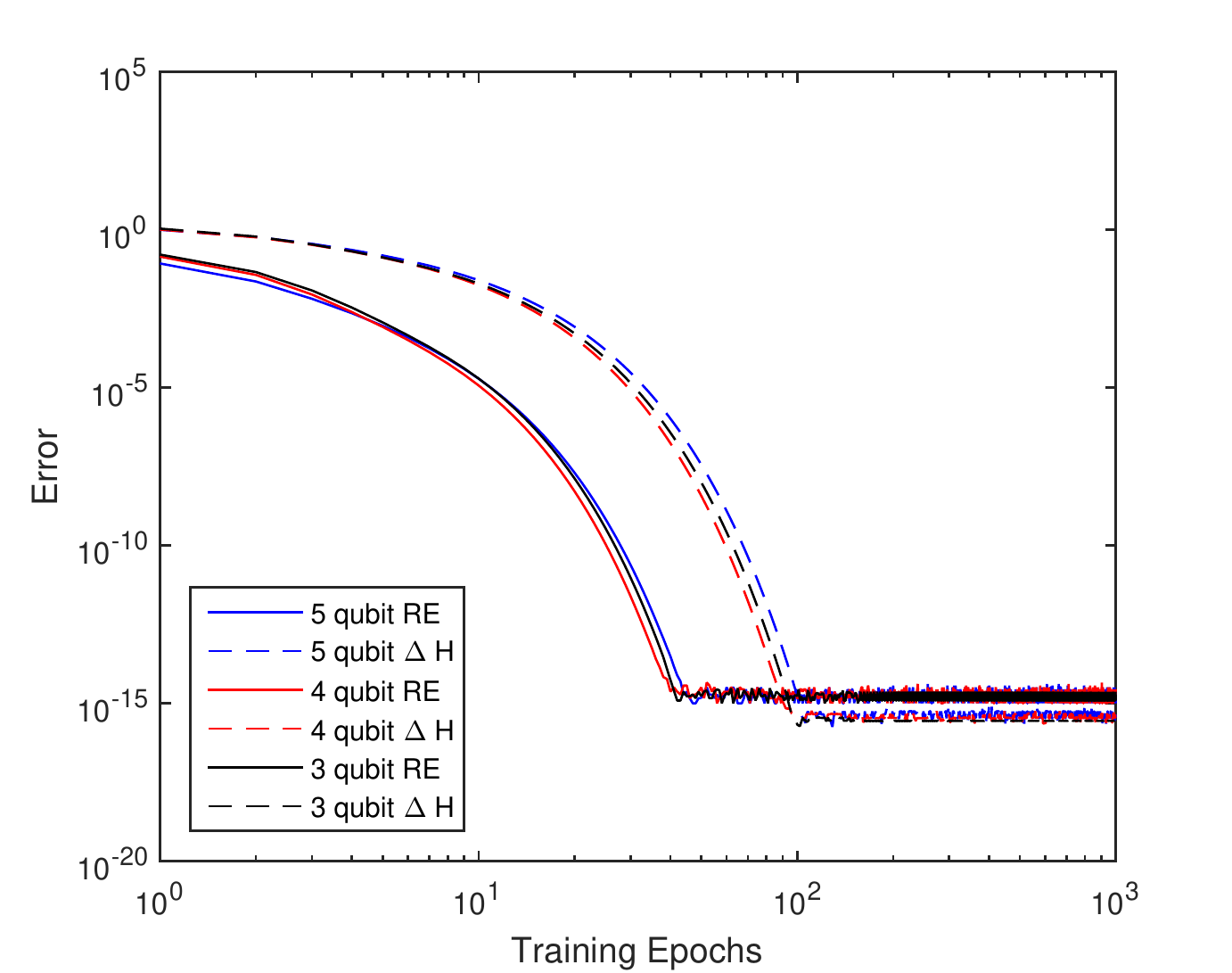}
\end{minipage}
\hspace{1mm}
\begin{minipage}{0.49\linewidth}
\includegraphics[width=\textwidth]{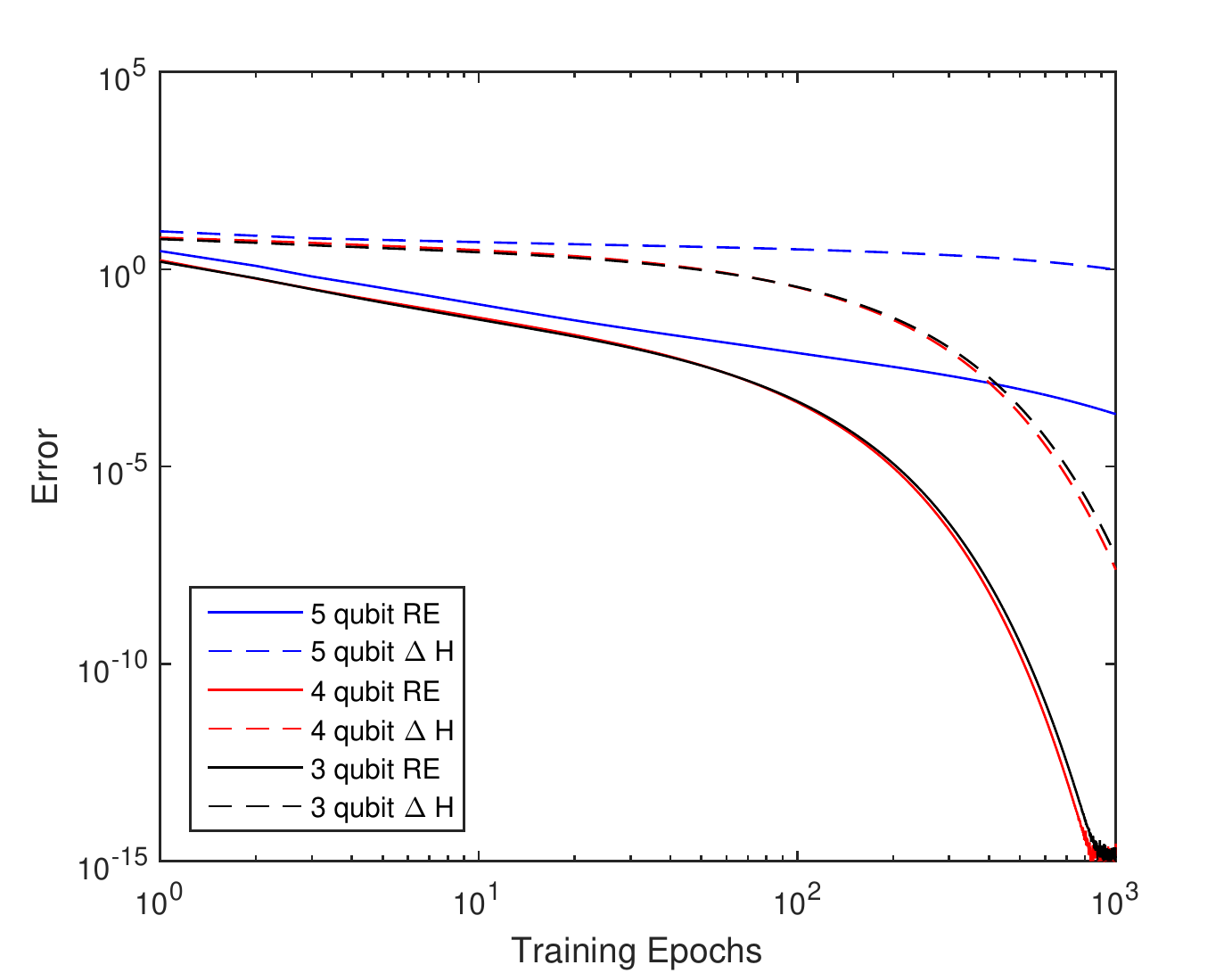}
\end{minipage}
\caption{Relative entropies and Hamiltonian errors for learning transverse Ising models.  The left figure shows data for a TI Hamiltonian with Gaussian random terms that is rescaled to unit norm.  The right figure shows the analogous situation but without normalizing the Hamiltonian.  Here $\Delta H= \|H_{\rm true} - H_{\rm est}\|_2$.\label{fig:Hlearn}}
\end{figure*}

\section{Relative Entropy Training}
In this section, we provide further numerical experiments that probe the performance of quantum relative entropy training.  The first that we consider is in~\fig{relPurity} which shows the performance of this form of training for learning randomly chosen $2$--qubit pure and mixed states.  In particular, we choose the pure states uniformly with respect to the Haar measure and pick the mixed states by generating the eigenvectors of Haar-random unitaries and choosing our mixed states to be convex combinations of such states with weights that are uniformly distributed.  

We see from these experiments that the median performance of relative entropy training on mixed states is quite good.  The quantum relative entropy is observed to shrink exponentially with the number of training epochs.  After as few as $35$ training epochs with $\eta=1$, the error is limited by numerical precision.  However, a glance at the $95\%$ confidence interval in this figure reveals that many of the examples yield much larger errors than these.  Specifically after $60$ epochs with the same learning rate the $97.5^{\rm th}$ percentile of the data in~\fig{relPurity} only has a relative entropy of $10^{-5}$ and is decaying much slower than the median.

The origin of this problem can be seen from the plot of the relative entropy for pure states in~\fig{relPurity}.  Pure states are observed to require many more training epochs to achieve the same accuracy as highly mixed states.  This is expected because pure states are only possible in the limit as $\|H\|\rightarrow \infty$.  The need to have large weights in the Hamiltonian not only means that more epochs will be needed to allow the weights to reach the magnitudes needed to approximate a pure state, but it also means that the training landscape is expected to be much more rough as we approach this limit.  This is what makes learning such pure states difficult.  Similarly, the fat tails of the error distribution for the mixed state case makes sense given that some of the data will come from nearly pure states.

The narrowing of error bars in these plots can be understood, approximately, from Levy's lemma.  Levy's lemma states that for any Lipschitz continuous function mapping the unit sphere in $2N-1$ dimensions (on which the pure states in $\mathbb{C}^{N}$ can be embedded) the probability that $f(x)$ deviates from its Haar expectation value by $\epsilon$ is in $e^{-O(N\epsilon^2)}$.  Thus if we take $f(x) = \bra{x} \sigma \ket{x}$, as we increase $N$ we expect almost all initial states $x$ chosen uniformly at random according to the Haar measure to have the widths of their confidence intervals in $O(1/\sqrt{N})\subseteq O(2^{-n/2})$, where $n$ is the number of qubits.  This means that the we expect the width of the confidence intervals to shrink exponentially with the number of qubits for cases where the target state is pure.  We do not necessarily expect similar concentrations to hold for mixed states because Levy's lemma does not directly apply in such cases.

\section{Applications to Hamiltonian Learning}
In all of the above applications our aim is to learn a Hamiltonian that parameterizes a thermal state model for the training data.  However, in some cases our aim may not be to learn a particular input state but to learn a system Hamiltonian for a thermalizing system.  Relative entropy training then allows such a Hamiltonian to be learned from the thermal expectation values of the Hamiltonian terms via gradient ascent and a simulator.  Here we illustrate this by moving away from a Hamiltonian model that is composed of a complete set of Pauli operators, to a local Hamiltonian model that lacks many of these terms.  Specifically, we choose a transverse Ising model on the complete graph:

\begin{equation}
H = \sum_j \alpha_j Z^j + \sum_j\beta_j X^j + \sum_{<i,j>} \gamma_{i,j} Z^{i} Z^j.\label{eq:HTI}
\end{equation}
We then test the ability of our training algorithm to reconstruct the true Hamiltonian given access to the requisite expectation values.

Apart from the simplicity of the transverse Ising model, it is also a useful example because in many cases these models can be simulated efficiently using quantum Monte-Carlo methods.  This means that quantum computers are not necessary for estimating gradients of models for large quantum systems.

\fig{Hlearn} shows that the ability to learn such models depends strongly on the norm of the Hamiltonian, or equivalently the inverse temperature of the thermal state.  It is much easier for us to learn a model using this method for a high temperature state than a low temperature thermal state.  The reason for this is similar to what we observed previously.  Gradient ascent takes many steps before it can get within the vicinity of the correct thermal state.  This is especially clear when we note that the error changes only modestly as we vary the number of qubits, however it changes dramatically when we vary the norm of the Hamiltonian.  This means that it takes many more training epochs to reach the region where the errors shrink exponentially from the initially chosen random Hamiltonian.  In cases where a good ansatz for the Hamiltonian is known, this process could be sped up.

\begin{figure*}[t!]
\begin{minipage}{0.49\linewidth}
\includegraphics[width=\textwidth]{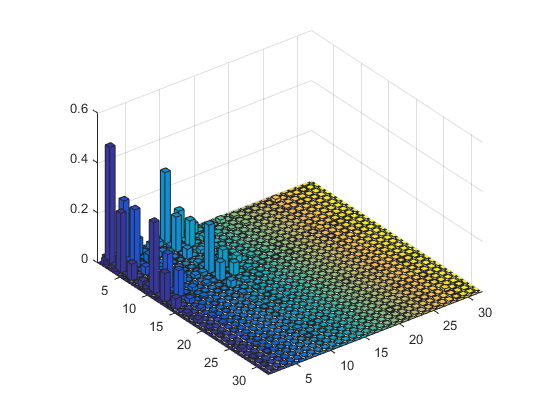}
(a) Thermal state for transverse Ising Model
\end{minipage}
\hspace{1mm}
\begin{minipage}{0.49\linewidth}
\includegraphics[width=\textwidth]{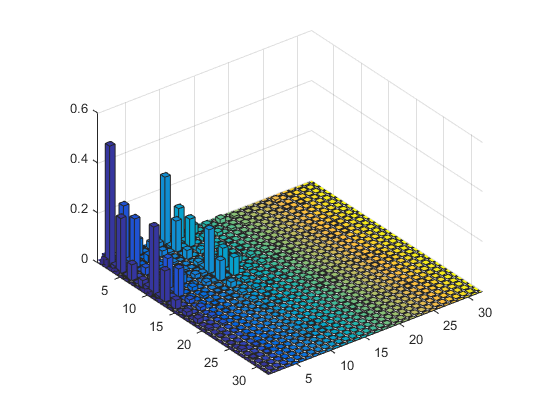}
(b) Mean-field approximation
\end{minipage}
\caption{Absolute value of mean field approximations to the thermal state of a $5$ qubit random TI Hamiltonian where each Hamiltonian term was chosen by sampling from a Gaussian with zero mean and unit variance at $\beta=1$.  The learning rate was taken to be $\eta=1$ and $100$ training epochs were used.\label{fig:MFBar}}
\end{figure*}

\begin{figure*}
\begin{minipage}{0.49\linewidth}
\includegraphics[width=\textwidth]{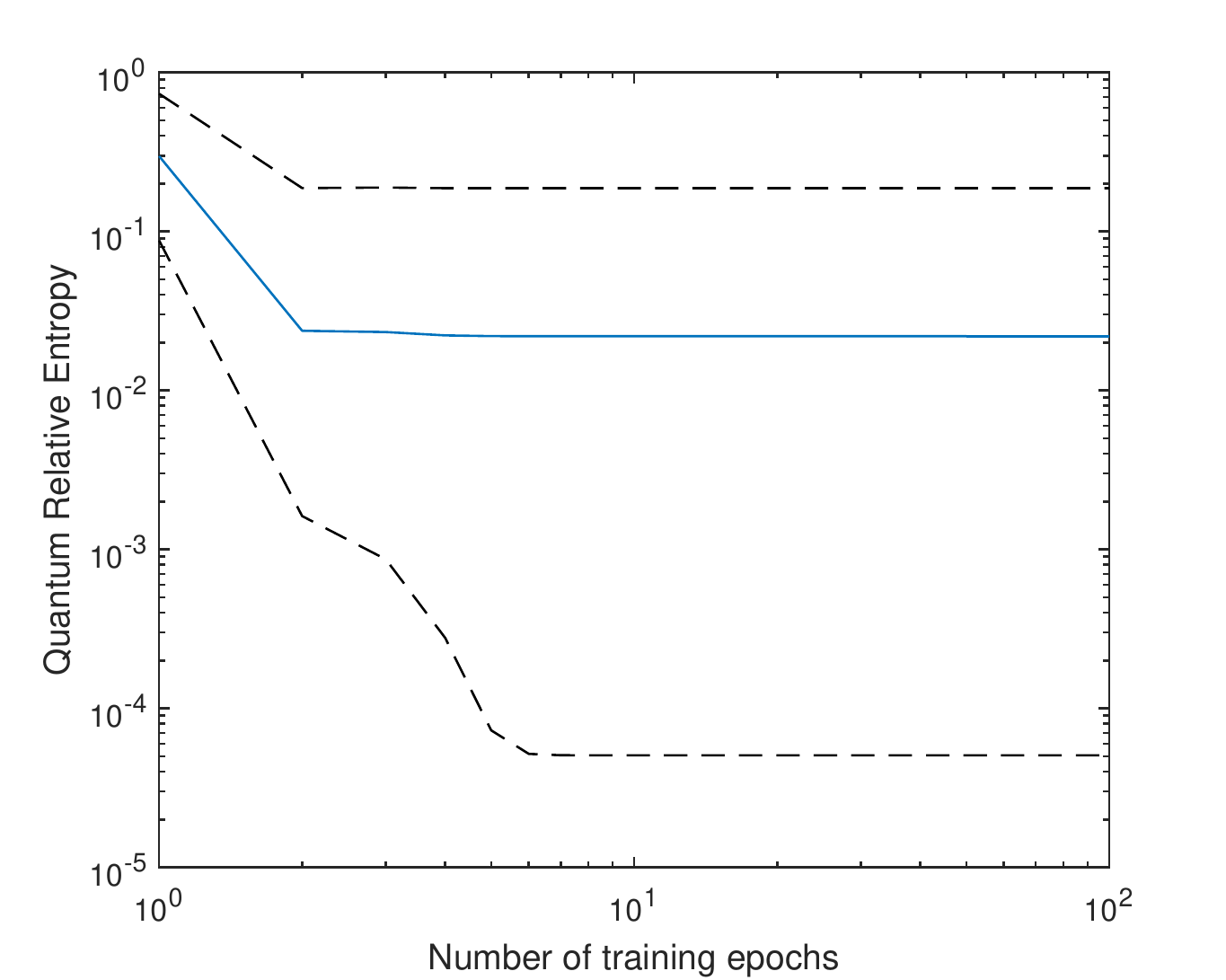}
\end{minipage}
\hspace{1mm}
\begin{minipage}{0.49\linewidth}
\includegraphics[width=\textwidth]{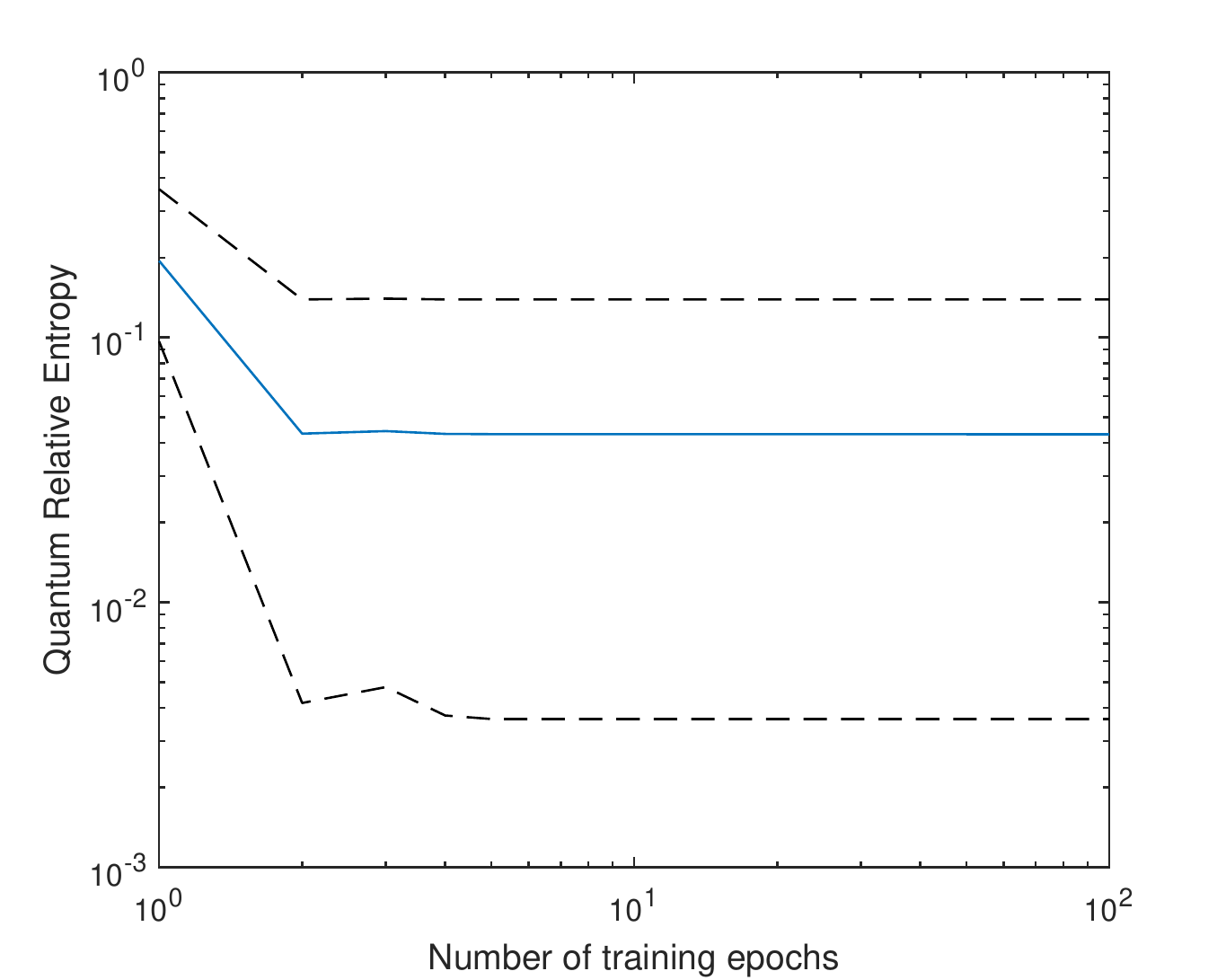}
\end{minipage}
\begin{minipage}{0.49\linewidth}
\includegraphics[width=\textwidth]{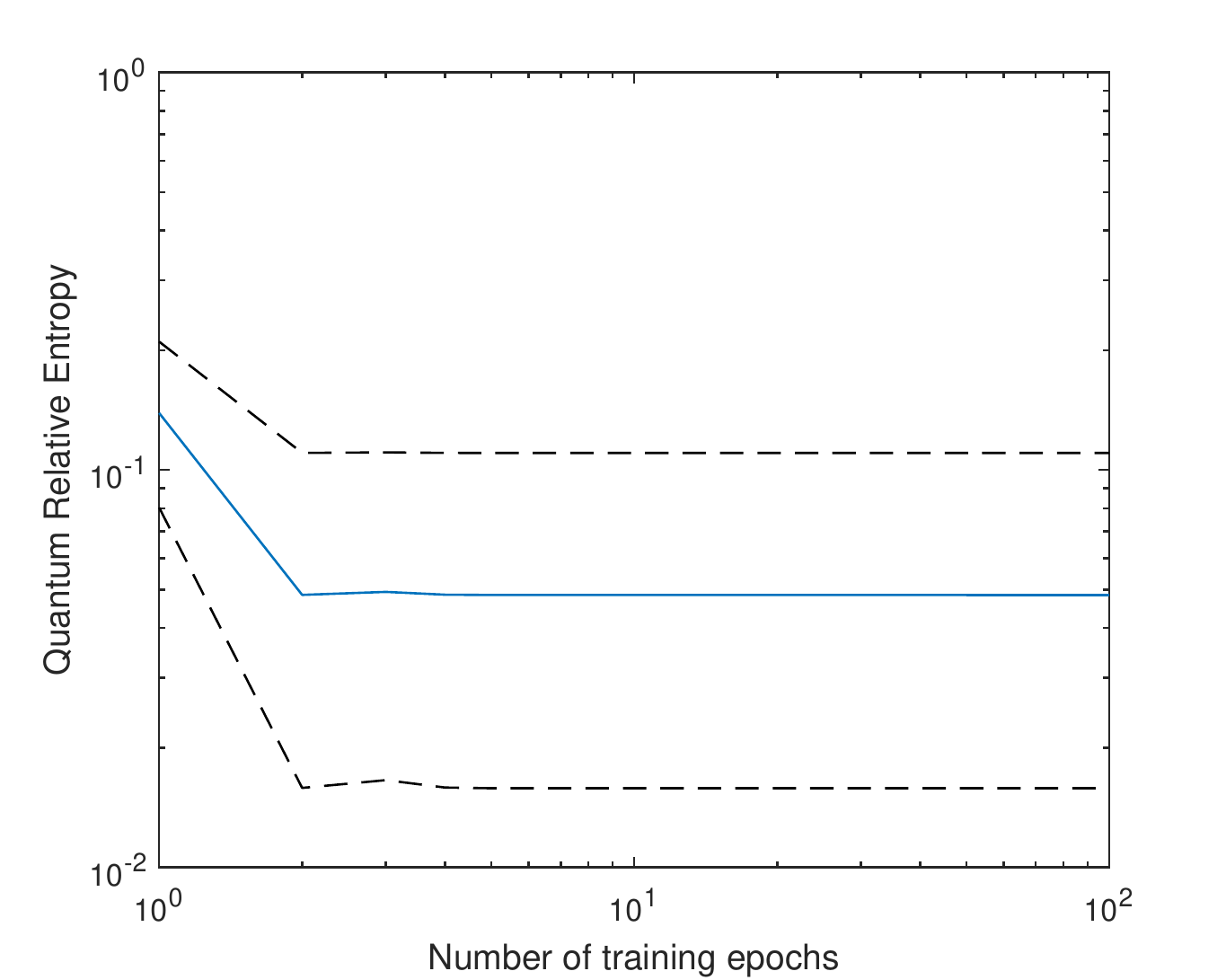}
\end{minipage}
\hspace{1mm}
\begin{minipage}{0.49\linewidth}
\includegraphics[width=\textwidth]{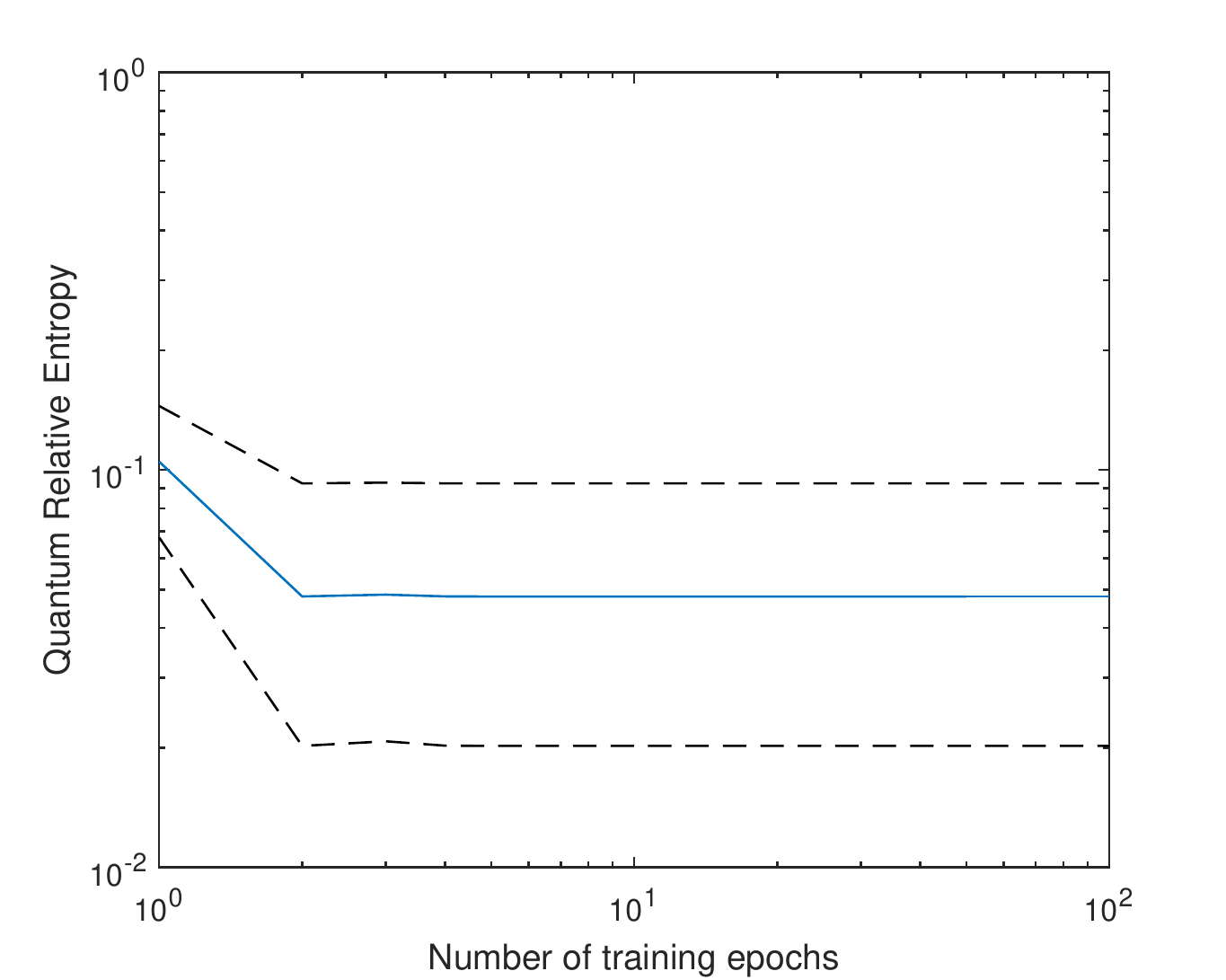}
\end{minipage}
\caption{Median relative entropies for mean-field and true distributions for thermal states generated by transverse Ising models on the complete graph with Gaussian random coefficients chosen with zero mean and unit variance for $2$ (top left), $3$ (top right), $4$ (bottom left) and $5$ (bottom right) qubits and $\eta=1$ was taken for each datum.  Dashed lines give a $95\%$ confidence interval. \label{fig:MF}}
\end{figure*}

\subsection{Mean Field Approximations}
Mean field approximations are ubiquitous in condensed matter physics.  They are relatively simple to compute for some quantum systems such as Ising models, but can be challenging for fully quantum models.  Here we provide a method to find a mean--field Hamiltonian for a system given the ability to compute moments of the density operator $\rho$.  The approach exactly follows the previous discussion, except rather than taking~\eq{HTI} we use
\begin{align}
H &= \sum_j H_j,\nonumber\\
H_j&:=\alpha_j Z^j + \beta_j X^j +\gamma_j Y^j.\label{eq:HTI}
\end{align}
Our aim is then to find vectors, $\alpha$, $\beta$ and $\gamma$ such that the correlated state $\rho$ is approximated by the uncorrelated mean field state:
\begin{equation}
\rho \approx e^{-H}/Z = \left[\prod_j e^{- H_j}\right]/Z.
\end{equation}

We see from the data in~\fig{MFBar} that relative entropy training on a thermal state that arises from a $5$ qubit transverse--Ising Hamiltonian on a complete graph for $100$ training epochs yields a mean field approximation that graphically is very close to the original state.  In fact if $\rho$ is the TI thermal state and $\sigma$ is the mean--field approximation to it then ${\rm Tr}(\rho\sigma)\approx 0.71$.  This shows that our method is a practical way to compute a mean field-approximation.

In order to assess how many epochs it takes in order to converge to a good mean-field approximation.  We see in~\fig{MF} that after only a single training epoch, the median relative entropy, over $1000$ randomly chosen instances, approximately reaches its optimal value.  Furthermore, we note that the relative entropy that the system saturates at tends to rise with the number of qubits.  This is in part due to the   fact that the Hamiltonian is on the complete graph and the weights are chosen according to a Gaussian distribution.  We therefore expect more correlated Hamiltonians as the number of qubits grows and in turn expect the mean-field approximation to be worse, which matches our observations.

If we turn our attention to learning mean--field approximations to $n$--local Hamiltonians for $n=2,\ldots,5$ we note that the mean-field approximation fails both qualitatively and quantitively to capture the correlations in the true distribution.  This is not surprising because such states are expected to be highly correlated and mean--field approximations should fail to describe them well.  These discrepancies continue even when we reduce the norm of the Hamiltonian.  This illustrates that the ability to find high--fidelity mean--field approximations depends less on the number of qubits than the properties of the underlying Hamiltonian.

\section{Commutator Training}

\begin{figure}[t!]
\includegraphics[width=\linewidth]{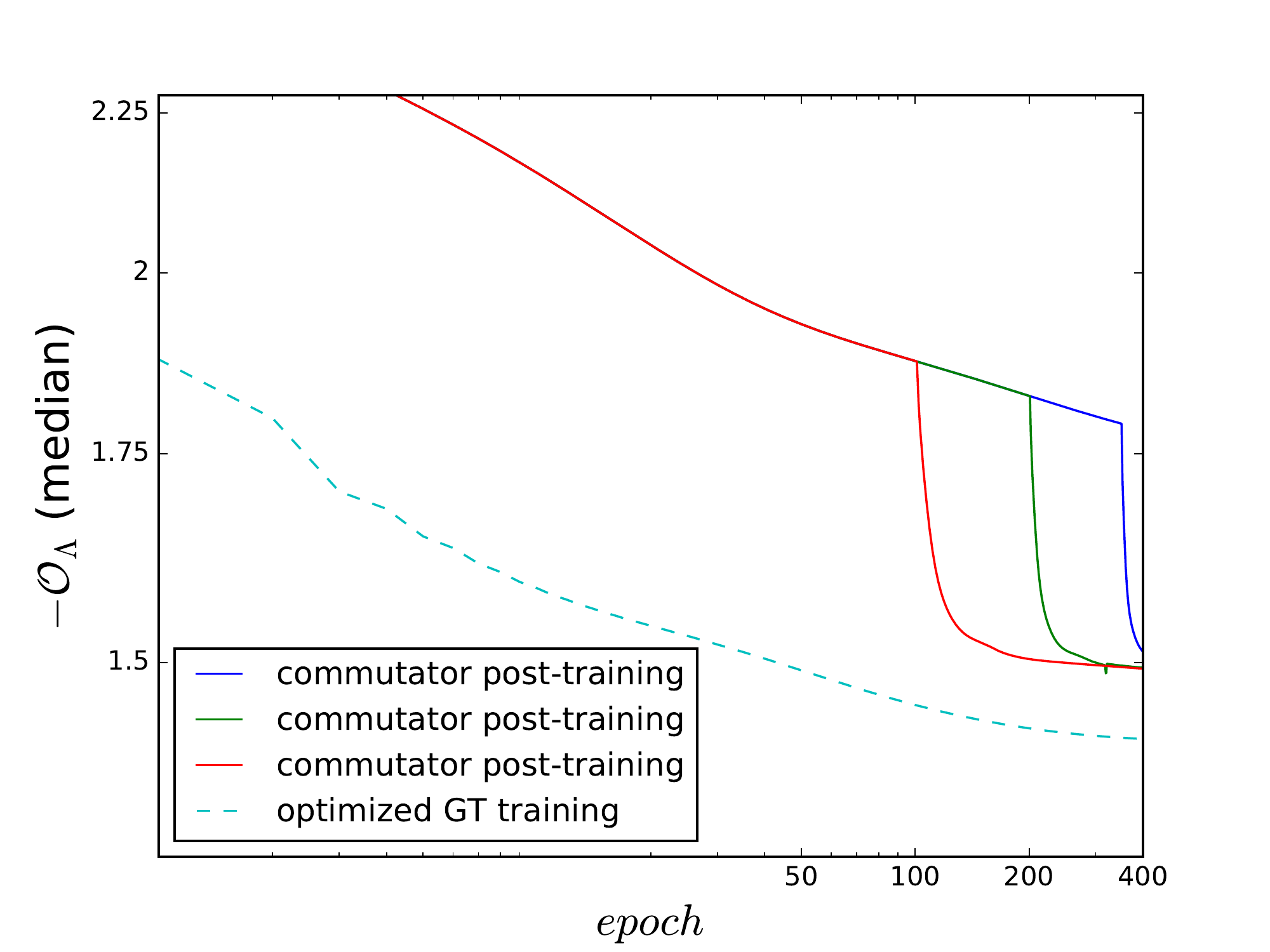}
\caption{Plot showing the efficacy of commutator training for all-visible Boltzmann machines with $4$ visible units. The top lines depict training with Golden-Thompson at 
first and then switching to commutator training where we see a sudden increase in accuracy. We picked the parameters such that the 
commutator training is stable. The bottom line (dashed) shows the performance of Golden-Thompson training with optimized learning rate and 
momentum.}\label{fig:commfig}
\end{figure}

\emph{Commutator training}-- A second approach that can be taken to avoid the use of the Golden--Thompson inequality.  The idea behind this approach is to approximate the
series in~\eq{duhamel} as a commutator series using Hadamard's Lemma.  In particular, if the Hamiltonian is a sum of bounded Hamiltonian terms then we have that~\eq{duhamel} becomes ${\rm Tr} [C e^{-H}]$ for
\begin{equation}
C:=\Lambda_v \left(\partial_\theta H+\frac{[H,\partial_\theta H]}{2!}+\frac{[H,[H,\partial_\theta H]]}{3!} + \cdots\right)
\end{equation}
Thus the gradient of the average log-likelihood becomes
\begin{align}
\sum_{\mathbf{v}}P_{\mathbf{v}} \Big( -\frac{\text{Tr}[e^{-H} C]} {\text{Tr}[e^{-H}] } + \frac{\text{Tr}[e^{-H}\partial_{\theta}H]} {\text{Tr}[e^{-H}] }  \Big)-\lambda h_{\theta} \delta_{H_\theta \in H_Q}.
\end{align}
This commutator series can be made tractable by truncating it at low order, which will not incur substantial error if $\|[H,\partial_\theta H]\|\ll 1$.  Commutator training is therefore expected to outperform Golden--Thompson training in the presence of $L_2$ regularization on the quantum terms, but is not as broadly applicable.

\begin{figure*}
\begin{minipage}{0.49\linewidth}
\includegraphics[width=\textwidth]{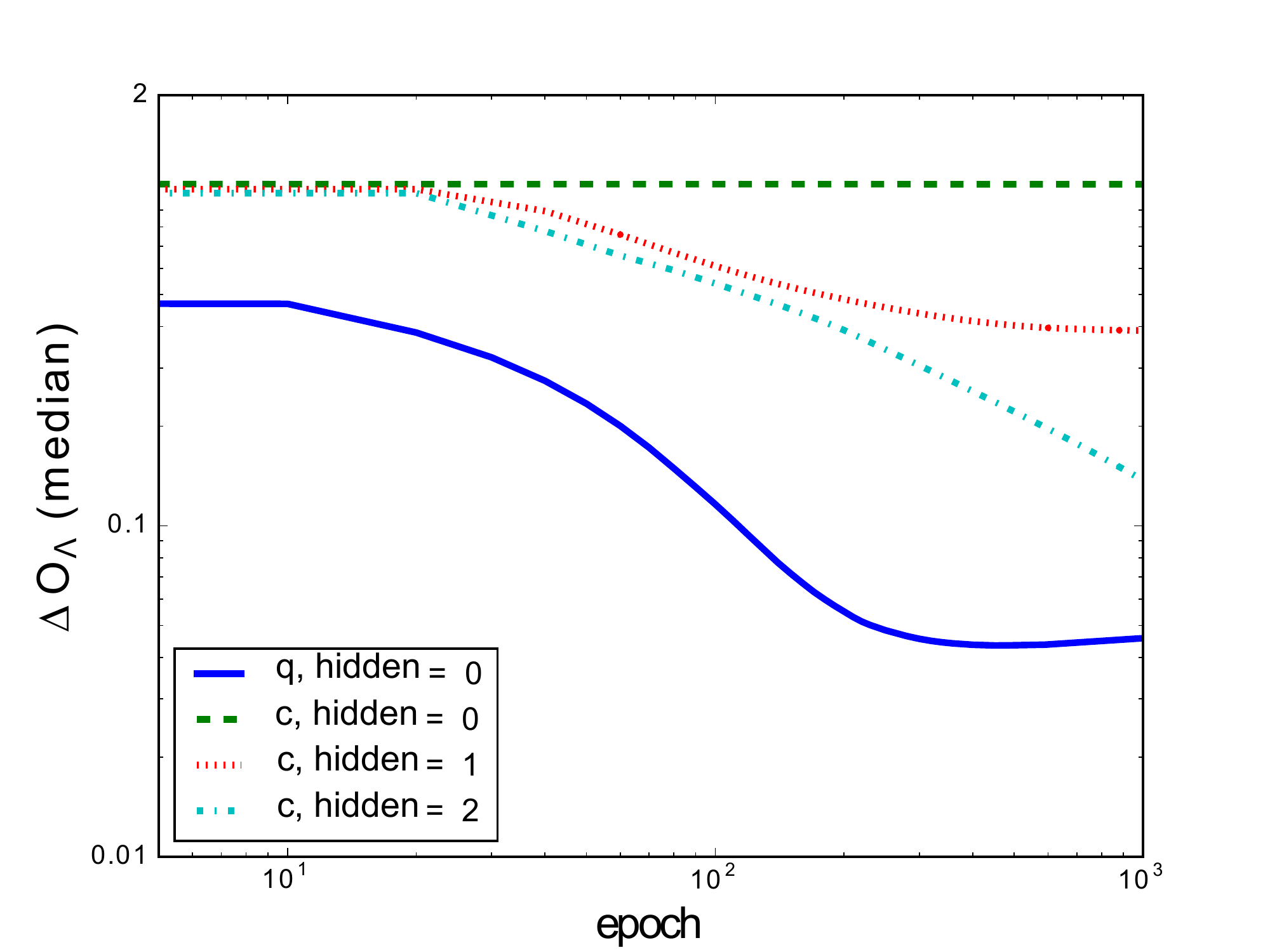}
(a)
\end{minipage}
\begin{minipage}{0.49\linewidth}
\includegraphics[width=\textwidth]{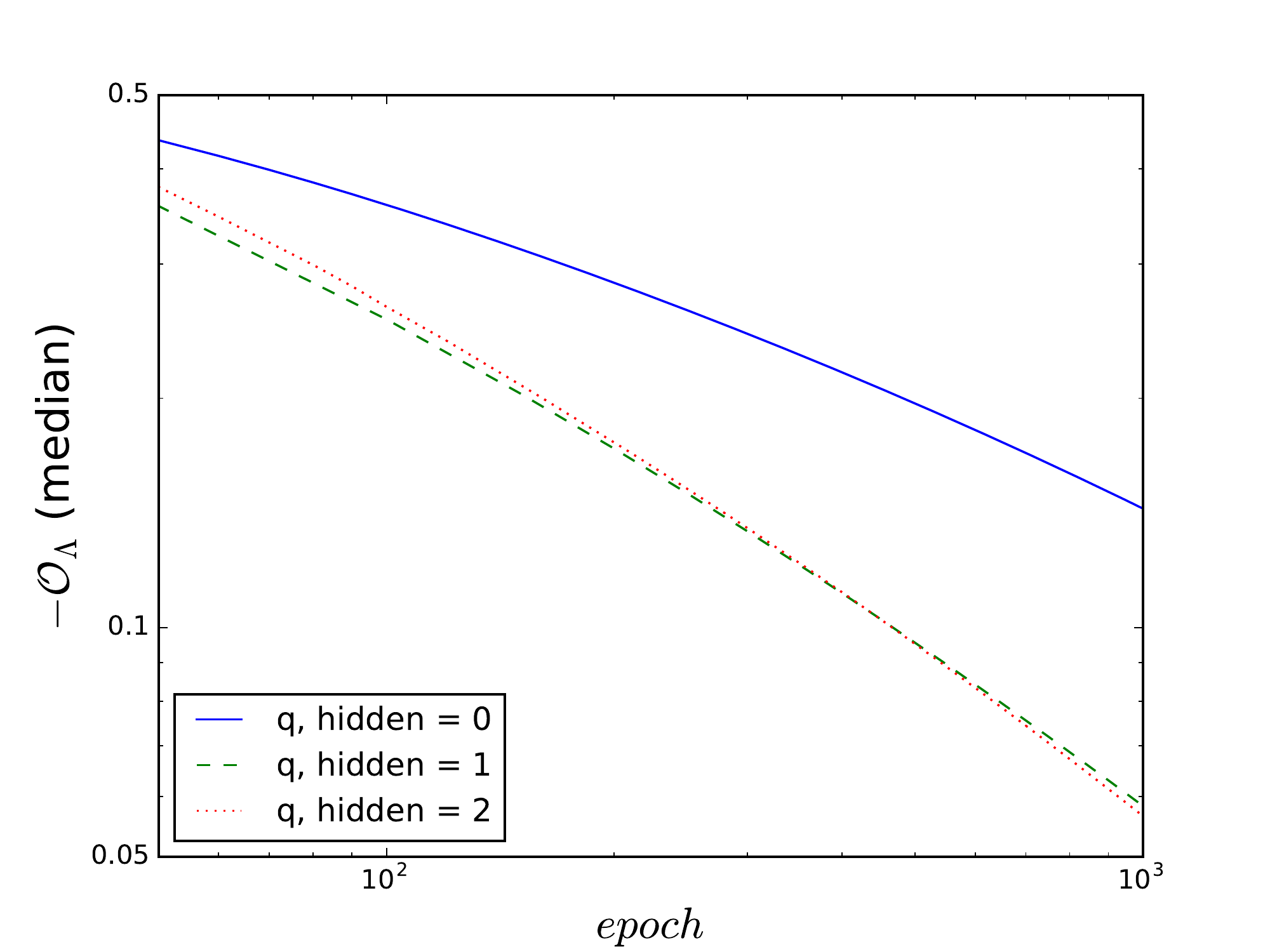}
(b)
\end{minipage}
\caption{Subfigure (a) shows the performance in 
terms of $\mathcal{O}_{\rho}$ for all-visible quantum Boltzmann machines and other parameters 
optimized for performance. In subfigure (b), we compare Boltzmann machines with 4 visible units and different number of hidden units.\label{fig:app2}}

\end{figure*}

We see in~\fig{commfig} that for a fixed learning rate that the gradients returned from a Golden-Thompson expansion are inferior to those returned from a high--order commutator expansion.  This in turn illustrates the gap between the exact gradients and Golden-Thompson gradients.  We examine this by performing Golden-Thompson trainings for an all-visible Boltzmann machine with $4$ visible units.  We train for a fixed number of epochs using the Golden-Thompson gradients and then switch to a $5^{\rm th}$--order commutator expansion.  We see a dramatic improvement in the objective function as a result.  This shows that in some circumstances much better gradients can be found with the commutator method than with Golden-Thompson; albeit at a higher price due to the fact that more expectation values need to be measured.

A drawback of the commutator method is that we find in numerical experiments that it is much less stable than Golden-Thompson.  In particular, commutator training does not fail gracefully when the expansion does not converge or when the learning rate is too large.  This means that the optimal learning rate for this form of training can substantially differ from the optimal learning rate for Golden-Thompson training.  When we optimize the learning rate for Golden--Thompson training we find that the training objective function increases by a factor of  roughly $1.5$, falling in line with the results seen using commutator training.  This shows that while commutator training can give more accurate gradients , it does not necessarily require fewer gradient steps.  In practice, the method is likely to be used in the last few training epochs after Golden--Thompson training, or other forms of approximate training, reach a local optima.

\section{Additional Experiments for POVM-based Training}
While the numerics in the main body provided a glimpse of the ability of POVM-based training to learn general Hamiltonian models, we provide a few additional experiments here to look at the performance of the training algorithm for different sizes of Fermionic Boltzmann machines.  We first examine the performance of the algorithm as a function of the number of hidden units for a $6$ visible unit example in~\fig{app2}.  We note here that while we can increase the number of hidden units in the classical model to help improve the objective function,

We see from~\fig{app2} that the inclusion of hidden units can have a dramatic improvement on the classical model's ability to learn.  In the quantum case we see that even the all-visible model outperforms each of the classical cases considered.  Adding a single hidden unit does substantially help for a $4$ visible unit model in the quantum case, but additional hidden units do not provide the quantum Boltzmann machine with much greater power for this training set.  This vindicates that the idea of deep learning still has a role for these quantum models despite the fact that the POVM is a projector onto a pure state and its compliment.  However, the lack of systematic improvements observed for larger instances suggest that the correlations present in the training data can be easily represented using the $H_{pqrs}$ terms present in the Fermionic Hamiltonian, since the impact of such terms is greatly diminished in the $4$ visible unit case.  More work will be needed in order to systematically study the role that hidden units play in deep learning for Fermionic Boltzmann machines and related models.

\section{Complexity Analysis}
We assume the following cost model here.  We assume that we have an oracle, $F_{H}(\epsilon_H)$, that is capable of taking the weights and biases of the quantum Boltzmann machine (or equivalently a parameterization of $H$) and outputs the state $\sigma$ such that $\|\sigma -e^{-H}/Z\| \le \epsilon_H$ for $\epsilon_H\ge 0$.  We manifestly assume that the state preparation is not exact because any computational model that grants the ability to prepare exact Gibbs states for arbitrary Hamiltonians is likely to be more powerful than quantum computing under reasonable complexity theoretic assumptions.  For relative entropy training, we also assume that the training data $\rho$ is provided by a query to an auxiliary oracle $F_{\rho}$.  We cost both oracles equivalently.  Finally, we assume for POVM training that the POVM elements can be prepared with a constant sized circuit and do not assign a cost to implementing such a term.  We do this for two reasons.  First for most elementary examples the POVM elements are very simple projectors and are not of substantially greater complexity than implementing a Hamiltonian term.  The second is that incorporating a cost for them would necessitate opening the blackbox $F_{H}$ which would substantially complicate our discussion and force us to specialize to particular state preparation methods.

The first result that we show is a lower bound based on tomographic bounds that shows that quantum Boltzmann training cannot be efficient in general if we wish to provide a highly accurate generative model for the training data.

\begin{lemma}
The number of queries to $F_{\rho}$ which yields copies of rank $r$ state operator $\rho\in \mathbb{C}^{D\times D}$ required to train an arbitrary quantum Boltzmann machine using relative entropy such that the quantum state generated by the Boltzmann machine are within trace distance $\epsilon\in (0,1)$ of $\rho$, and with failure probability $\Theta(1)$, is in $\Omega(Dr/[\epsilon^2\log(D/r\epsilon)])$.
\end{lemma}
\begin{proof}
The proof follows by contradiction.  Since we have assumed an arbitrary quantum Boltzmann machine we will consider a Boltzmann machine that has a complete set of Hamiltonian terms.  If we do not make this assumption then there will be certain density operators that cannot be prepared within error $\epsilon$ for all $\epsilon>0$.  Let us assume that $\rho$ is rank $D$ if this is true then there exists $H\in \mathbb{C}^{D\times D}$ such that $\rho \propto e^{-H}$ because the matrix logarithm is well defined for such systems.  

Now let us assume that $\rho$ has rank less than $D$.  If that is the case then there does not exist $H\in \mathbb{C}^{D\times D}$ such that $\rho \propto e^{-H}$, but $\rho$ can be closely approximated by it.  Let $P_0$ be a projector onto the null space of $\rho$, which we assume is $D-r$ dimensional.  Then let $\tilde{\rho}\in\mathbb{C}^{r\times r}$ be the projection of $\rho$ onto the orthogonal compliment of its null space.  Since $\rho$ is maximum rank within this subspace, there exists $\tilde{H}\in\mathbb{C}^{r\times r}$ such that $\tilde{\rho} \propto e^{-\tilde{H}}$.
After a trivial isometric extension of $\tilde{H}$ to $\mathbb{C}^{D\times D}$, we can then write $\rho \propto (\openone -P_0)e^{-\tilde{H}}(\openone -P_0)$.  By construction $[\tilde{H},(\openone -P_0)]=0$, and thus $\rho \propto (\openone -P_0)e^{-\tilde{H}}= (\openone -P_0)e^{-(\openone -P_0)\tilde{H}(\openone -P_0)}$.  

The definition of the trace norm implies that for any $\gamma>0$, $\|(\openone - P_0) - e^{-\gamma P_0}\|_{1}\in O([D-r]e^{-\gamma})$.  Thus because $e^{-(\openone -P_0)\tilde{H}(\openone -P_0)}/Z$ has trace norm $1$
\begin{align}
\rho& = e^{-\gamma P_0}e^{-(\openone -P_0)\tilde{H}(\openone -P_0)}/Z +O([D-r]e^{-\gamma})\nonumber\\
&=e^{-(\openone -P_0)\tilde{H}(\openone -P_0)-\gamma P_0}/Z +O([D-r]e^{-\gamma}).
\end{align}
Thus $\rho$ can be approximated within error less than $\epsilon$, regardless of its rank, by a Hermitian matrix whose norm scales at most as $O(\|\tilde{H}\|+\log(D/\epsilon))$.  Thus for every $\epsilon>0$ there exists a quantum Boltzmann machine with a complete set of Hamiltonian terms that can approximate $\rho$ within trace distance less than $\epsilon$ using a bounded Hamiltonian.

Haah, Harrow et al show in Theorem 1 of~\cite{haah2015sample} that $\Omega(Dr/[\epsilon^2\log(D/r\epsilon)])$ samples are needed to tomographically reconstruct a rank $r$ density operator $\rho \in \mathbb{C}^{D\times D}$ within error $\epsilon$ in the trace distance.  Since training a Boltzmann machine can provide a specification of an arbitrary density matrix, to within trace distance $\epsilon$, if this training process required $\omega(Dr/[\epsilon^2\log(D/r\epsilon)])$ samples we would violate their lower bound on tomography.  The result therefore follows.
\end{proof}

\begin{lemma}
There does not exist a general purpose POVM-based training algorithm for quantum Boltzmann machines on a training set such that $|\{P_v: P_v >0\}|=N$ can prepare a thermal state such that
${\rm Tr}([\sum_v P_v \Lambda_v] e^{-H} /Z) \ge 1/\Delta$ that requires $M$ queries to $P_v$ where $\Delta \in o(\sqrt{N})$ and $M\in o(\sqrt{N})$. 
\end{lemma}
\begin{proof}
The proof naturally follows from reducing Grover's search to Boltzmann training.  We aim to use queries to the blackbox oracle to learn a whitebox oracle that we can query to
learn the marked state without actually querying the original box.  To be clear, let us pick $\Lambda_0 = \ketbra{0}{0}$ and $P_1=1$ and for $v>1$, $\Lambda_v = \ketbra{v}{v}$ with $P_v=0$.  These elements form a POVM because they are positive and sum to the identity.

In the above construction the oracle that gives the $P_v$ is equivalent to the Grover oracle.  This implies that a query to this oracle is the same as a query to Grover's oracle.

Now let us assume that we can train a Boltzmann machine such that ${\rm Tr}(\Lambda_0 e^{-H}/Z)\in \omega(1/\sqrt{N})$ using $o(\sqrt{N})$ queries to the blackbox.  This implies that $o(\sqrt{N})$ queries are needed on average to prepare $\ket{0}$ by drawing samples from the BM and verifying them using the oracle.  Since the cost of learning the BM is also $o(\sqrt{N})$, this implies that the number of queries needed in total is $o(\sqrt{N})$.  Thus we can perform quantum search under these assumptions using $o(\sqrt{N})$ queries and hence from lower bounds this implies $o(\sqrt{N}) \subseteq \Omega(\sqrt{N})$ which is a contradiction.
\end{proof}

The above lemmas preclude general efficient Boltzmann training without further assumptions about the training data, or without making less onerous requirements on the precision of the BM model output by the training algorithm.  This means that we cannot expect even quantum Boltzmann machines have important limitations that need to be considered when we examine the complexity of quantum machine learning algorithms.

\begin{theorem}
Let $H=\sum_{j=1}^M \theta_j H_j$ with $\|H_j\|=1~\forall~j$ be the Hamiltonian for a quantum Boltzmann machine and let $G$ be an approximation to $\nabla \mathcal{O}$ where $\mathcal{O}$ that is the training objective function for either POVM based or relative entropy training.  There exist training algorithms such that at each of the $N_{\rm epoch}$ epochs $\mathbb{E}(\|G - G_{\rm true}\|_2^2)\le \epsilon^2$ and query $F_H$ and the training set a number of times that is in $$O\left(N_{\rm epoch} \left(\frac{M^2}{\epsilon^2} \right)\right).$$
\end{theorem}
\begin{proof}
We show the proof by considering the approximate gradients given by the methods in the main body.  Since each of those methods uses sampling the result will follow from straight forward estimates of the variance.  Consider an unbiased estimator of the mean such as the sample mean.  Since such an estimate is unbiased it satisfies $\mathbb{E}(G)  = G_{\rm true}$.  Thus \begin{align}
\mathbb{E}(\|G-G_{\rm true}\|_2^2) &= \mathbb{E}(\|G-\mathbb{E}(G)\|_2^2)\nonumber\\
&=\sum_{j=1}^M \mathbb{E}\left((G_j-\mathbb{E}(G_j))^2\right)\nonumber\\
&=\sum_{j=1}^M \mathbb{V}\left(G_j\right).
\end{align}
For relative entropy training under the assumption that $\|H_j\|\le 1$ for all $j$
\begin{equation}
\mathbb{V}(G_j) \in O(\max\{{\rm Tr}(\rho H_j), {\rm Tr}(H_j e^{-H}/Z) \}/n)\le 1/n.
\end{equation}
Similarly for POVM training
\begin{align}
\mathbb{V}(G_j) &\in O(\max\{{\rm Tr}\big(H_j\sum_v \frac{P_v e^{-H_v}}{Z_v}), {\rm Tr}(H_j \frac{e^{-H}}{Z_v}\big)\}/n)
\nonumber\\
&\in O(1/n).
\end{align}
Therefore
\begin{equation}
\sum_{j=1}^M \mathbb{V}\left(G_j\right) \in O(M/n).
\end{equation}
Thus if we wish to take the overall variance to be $\epsilon^2$ it suffices to take $n=M/\epsilon^2$.  Each of these $n$ samples requires a single preparation of a thermal state and or a query to the training data.  Thus for both training algorithms considered the number of queries needed to compute a component of the gradient is $O(n)$.  Since there are $M$ components the total number of queries needed to estimate the gradient is in
\begin{equation}
O(nM) \subseteq O(M^2/\epsilon^2).
\end{equation}
The result then follows from the assumption that the algorithm makes $N_{\rm epoch}$ gradient steps.
\end{proof}

\bibliographystyle{unsrt}

\end{document}